\documentclass[11pt]{article}
\usepackage[pdftex,bookmarksopen=true,bookmarks=true,unicode,setpagesize]{hyperref}
\hypersetup{colorlinks=true,linkcolor=black,citecolor=black}

\usepackage{amssymb, amsmath, amsthm}

\allowdisplaybreaks
\usepackage{cite}

\newcommand\R{\mathbb{R}}

\newtheorem{theorem}{Theorem}[section]
\newtheorem{corollary}[theorem]{Corollary}
\newtheorem{lemma}[theorem]{Lemma}
\newtheorem{proposition}[theorem]{Proposition}

\theoremstyle{remark}
\newtheorem{definition}[theorem]{Definition}
\theoremstyle{remark}
\newtheorem{example}[theorem]{Example}
\theoremstyle{remark}
\newtheorem{remark}[theorem]{Remark}
\newtheorem{assumptions}[theorem]{Assumption}
\theoremstyle{remark}

\begin{document}

\begin{center}{\Large \bf
Quasi-free states on a class of algebras of multicomponent commutation relations}
\end{center}

{\large Eugene Lytvynov\\ Department of Mathematics, Swansea University,  Bay Campus, Swansea, SA1 8EN, UK;\\
e-mail: \texttt{e.lytvynov@swansea.ac.uk}\vspace{2mm}

{\large Nedal Othman}\\ 
Dubai Women's College, Higher Colleges of Technology,  Baghdad Street, Al Nahda, Dubai, P.O. Box: 16062, The United Arab Emirates;\\
e-mail: \texttt{nothman@hct.ac.ae}\vspace{2mm}

{\small

\begin{center}
{\bf Abstract}
\end{center}

\noindent Multicomponent commutations relations (MCR) describe  plektons, i.e., multicomponent quantum systems   with a generalized statistics.  In such systems, exchange of quasiparticles is governed by a unitary matrix $Q(x_1,x_2)$ that depends on the position of quasiparticles. For such an exchange to be possible, the matrix must satisfy several conditions, including the functional Yang--Baxter equation. The aim of the paper is to give an appropriate definition of a quasi-free state on an MCR algebra,  and construct such states on a class of  MCR algebras. We observe a significant difference between the classical setting for bosons and fermions and the setting of MCR algebras. We show that the developed theory is applicable to systems that contain quasiparticles of opposite type. An example of such a system is  a two-component system in which two quasiparticles, under exchange, change their respective types to the opposite ones ($1\mapsto 2$, $2\mapsto1$). Fusion of quasiparticles means intuitively putting several quasiparticles in an infinitely small box and identifying the statistical behaviour of the box. By carrying out fusion of an odd number of particles from the two-component system as described above, we obtain further examples of quantum systems to which the developed theory is applicable.

 } \vspace{2mm}

\noindent
{\bf Keywords:} Anyon; Plekton; Fusion of quasiparticles; Quasi-free state
\vspace{2mm}

\noindent
{\bf Mathematics Subject Classification (2020):} Primary 81R10, 81V27; Secondary 47L90.    

\section{Introduction and preliminaries}

Multicomponent commutations relations (MCR) describe  plektons, i.e., multicomponent quantum systems   with a generalized statistics.
 In such systems, exchange of quasiparticles is governed by a unitary matrix that depends on the position of quasiparticles. For such an exchange to be possible, the matrix must satisfy several conditions, including the functional Yang--Baxter equation. 

The aim of the paper is to give an appropriate definition of a quasi-free state on an MCR algebra  and construct such states on a class of MCR algebras. We observe a significant difference between the classical setting for bosons and fermions and the setting of MCR algebras.

\subsection{Multicomponent commutation relations} \label{cydyccdfe}

The first paper pointing out the possibility of a multicomponent quantum system was the comment by Menikoff, Sharp and Goldin \cite{GoldinMenikoffSharp}. Such systems were rigorously   derived and studied by Liguori and Mintchev  \cite{LM}, see also Goldin and Majid  \cite{GoMa}. The paper  \cite{DKLP} gave an overview of multicomponent quantum systems with concrete  examples when the number of components of a quantum system is two. That paper actually treated a more general setting than the one considered in \cite{GoMa,LM}. For other results  related to multicomponent quantum systems, see e.g.\ Bo\.zejko, Speicher~\cite{BozSpe} and J{\o}rgensen, Schmitt, Werner \cite{jorgensenschmittwernerJFA1995}.

Let us briefly recall the definition of the MCR, for more detail see e.g.\ \cite{DKLP,GoMa,LM}. Let $X:=\R^d$ with $d\ge 2$ (the physically important case being $d=2$) and let $V:=\mathbb C^r$, $r$ being the number of components (types of quisiparticles) in the quantum system. Let $\{e_1,\dots,e_r\}$ be the standard orthonormal basis in $V$ (the $i$th coordinate of $e_i$ being 1, the other coordinates being 0). Let $J$ be the complex conjugation in $V$, i.e., the antilinear operator satisfying $Je_i=e_i$ for all $i$. Let $\mathcal L(V^{\otimes 2})$ denote the space of all linear operators (matrices) in $V^{\otimes 2}$. We fix a map $Q:\R^2\to\mathcal L(V^{\otimes 2})$ such that, for each $(y_1,y_2)\in \R^2$, $Q(y_1,y_2)$ is a unitary operator, $Q^*(y_1,y_2)=Q(y_2,y_1)$, and the functional Yang--Baxter equation is satisfied  in $V^{\otimes 3}$:
\begin{equation}\label{v6w6u34}
Q_1(y_1,y_2)Q_2(y_1,y_3)Q_1(y_2,y_3)=Q_2(y_2,y_3)Q_1(y_1,y_3)Q_2(y_1,y_2)\end{equation}
for $(y_1,y_2,y_3)\in \R^3$. In \eqref{v6w6u34} and below, for a linear operator $\mathcal C$ acting in $V^{\otimes 2}$ and $k\ge3$, we denote by $\mathcal C_i$ the linear operator in $V^{\otimes k}$ acting by $\mathcal C$ on the 
 $i$th and $(i+1)$th components of $V^{\otimes k}$. We also define, with an abuse of notation, for $x_i=(x_i^1,\dots,x_i^d)\in X$ ($i=1,2$),  $Q(x_1,x_2):=Q(x_1^1,x_2^1)$.

Let $\mathcal H:=L^2(X;V)$ be the $L^2$-space (with respect to the Lebesgue measure on $X$) of $V$-valued functions on $X$. We keep the notation $J$ for the complex conjugation in $\mathcal H$. Let  $A^+(f)$, $A^-(f)$  ($f\in\mathcal H$) be linear operators acting in a separable Hilbert space $\mathfrak F$, with a dense domain $\mathfrak D\subset \mathfrak F$. 
We assume that the operators $A^+(f)$, $A^-(f)$ depend linearly on $f$, $A^+(f)^*\restriction_{\mathfrak D}=A^-(Jf)$, and the operators $A^+(f)$, $A^-(f)$ map $\mathfrak D$ into itself. We introduce operator-valued distributions
\begin{equation}\label{vftye6u}
A^+(x)=(A_1^+(x),\dots,A^+_r(x)),\quad A^-(x)=(A_1^-(x),\dots,A^-_r(x))\quad (x\in X)\end{equation}
that satisfy, for all $f(x)=\sum_{k=1}^r f^k(x)e_k\in\mathcal H$,
\begin{equation}\label{vctesay5ra5}
 A^\sharp(f)=\sum_{k=1}^r\int_X f^k(x)A_k^\sharp(x)\,dx,\quad \sharp\in\{+,-\}.\end{equation}
For $u,v\in V$, denote $\langle u,v\rangle_V:=(u,Jv)_V=\sum_{k=1}^r u^kv^k$.
Then we may write formula~\eqref{vctesay5ra5} in the heuristic form $A^\sharp(f)=\int_X \langle f(x),A^\sharp(x)\rangle_V\,dx$. Analogously, we write, for a product of two operators,
\begin{align*}
A^{\sharp_1}(f_1) A^{\sharp_2}(f_2)=\int_{X^2}\big\langle f_1(x_1)\otimes f_2(x_2),A^{\sharp_1}(x_1)\otimes A^{\sharp_2}(x_2)\big\rangle_{V^{\otimes 2}}\, dx_1\, dx_2,\quad \sharp_1,\sharp_2\in\{+,-\}.
\end{align*}

We say that the operators $A^+(f)$, $A^-(f)$ satisfy the $Q$-MCR  if 
\begin{align}
&\int_{X^2}\big\langle f_1(x_1)\otimes f_2(x_2),A^+(x_1)\otimes A^+(x_2)\big\rangle _{V^{\otimes 2}}\,dx_1\,dx_2\notag\\
&\quad =
\int_{X^2}\big\langle Q(x_2,x_1) f_1(x_1)\otimes f_2(x_2),A^+(x_2)\otimes A^+(x_1)\big\rangle_{V^{\otimes 2}}\,dx_1\,dx_2,\label{xseas5aq53q}\\
&\int_{X^2}\big\langle f_1(x_1)\otimes f_2(x_2),A^-(x_1)\otimes A^-(x_2)\big\rangle _{V^{\otimes 2}}\,dx_1\,dx_2\notag\\
&\quad =
\int_{X^2}\big\langle \widehat Q(x_2,x_1) f_1(x_1)\otimes f_2(x_2),A^-(x_2)\otimes A^-(x_1)\big\rangle _{V^{\otimes 2}}\,dx_1\,dx_2,\label{buyfr86r68}\\
&\int_{X^2}\big\langle f_1(x_1)\otimes f_2(x_2),A^-(x_1)\otimes A^+(x_2)\big\rangle _{V^{\otimes 2}}\,dx_1\,dx_2 = \int_X \langle f_1(x),f_2(x)\rangle _V\,dx\notag\\
&\qquad +
\int_{X^2}\big\langle \widetilde Q(x_1,x_2) f_1(x_1)\otimes f_2(x_2),A^+(x_2)\otimes A^-(x_1)\big\rangle _{V^{\otimes 2}}\,dx_1\,dx_2.\label{yqdqidi}
\end{align}
In these formulas, $\widehat Q(x_1,x_2):=\mathbb S^{(2)}\, Q(x_2,x_1)\,\mathbb S^{(2)}$, where the antilinear operator $\mathbb S^{(2)}$ in $V^{\otimes 2}$ is given by 
$\mathbb S^{(2)}(u\otimes v):=(Jv)\otimes (Ju)$, 
and $\widetilde Q(x_1,x_2)$ satisfies 
\begin{equation}\label{tydse5asw}
\big\langle \widetilde Q(x_1,x_2)e_i\otimes e_j,e_k\otimes e_l\big\rangle _{V^{\otimes 2}}=\big\langle Q(x_1,x_2)e_k \otimes e_i,e_l\otimes e_j\big\rangle _{V^{\otimes 2}}\quad \forall i,j,k,l\in\{1,\dots,r\}. 
\end{equation}
Note that  the double integrals appearing on the right hand side of formulas \eqref{xseas5aq53q}--\eqref{yqdqidi} are assumed to be well-defined. Below, in Section~\ref{fgdyjdde}, we will recall the Fock representation of the $Q$-MCR.

Note that the $Q$-MCR \eqref{xseas5aq53q}--\eqref{yqdqidi} can be written in the following shorthand form:
\begin{align}
A^+(x_1)\otimes A^+(x_2)&=Q(x_2,x_1)^TA^+(x_2)\otimes A^+(x_1),\notag\\ 
A^-(x_1)\otimes A^-(x_2)&=\widehat Q(x_2,x_1)^TA^-(x_2)\otimes A^-(x_1),\notag\\
A^-(x_1)\otimes A^+(x_2)&=\delta(x_1-x_2)\operatorname{Tr}(\cdot)+\widetilde Q(x_1,x_2)^TA^+(x_2)\otimes A^-(x_1).\label{ydsydq}
\end{align}
Here, for $M\in\mathcal L(V^{\otimes 2})$, $M^T$ is the transposed of $M$; for $v^{(2)}\in V^{\otimes 2}$, $\operatorname{Tr}(v^{(2)}):=\sum_{k=1}^r (v^{(2)},e_k\otimes e_k)_{V^{\otimes 2}}$ is the trace\footnote{Note that, in the usual  way, $V^{\otimes 2}$ can be identified with $\mathcal L(V)$,  and then  $\operatorname{Tr}(v^{(2)})$ becomes  the   trace of the linear operator $v^{(2)}$.} of $v^{(2)}$; and for $f_1,f_2\in L^2(X;\mathbb C)$,
$$\int_{X^2} f_1(x_1)f_2(x_2)\delta(x_1-x_2)dx_1\,dx_2:=\int_X  f_1(x)f_2(x)dx.$$

When $r=1$ (i.e., the quantum system has a single component), $Q(x_1,x_2)$ is just a complex-valued function satisfying $|Q(x_1,x_2)|=1$ and $\overline{Q(x_1,x_2)}=Q(x_2,x_1)$ while equation~\eqref{v6w6u34} is trivially  satisfied. Furthermore, $\widehat Q(x_1,x_2)=Q(x_1,x_2)$ and $\widetilde Q(x_1,x_2)=Q(x_1,x_2)$. Hence, formulas \eqref{ydsydq} become
\begin{gather}
A^+(x_1)A^+(x_2)=Q(x_2,x_1)A^+(x_2)A^+(x_1),\quad A^-(x_1)A^-(x_2)=Q(x_2,x_1)A^-(x_2)A^-(x_1),\notag\\
A^-(x_1)  A^+(x_2)=\delta(x_1-x_2)+Q(x_1,x_2)A^+(x_2) A^-(x_1).\label{gdstsj6e5}
\end{gather}
We will call formulas \eqref{gdstsj6e5} the $Q$-anyon commutation relations ($Q$-ACR), see e.g.\ \cite{LM,BLW} and the references in \cite{DKLP}. 
For $Q(x_1,x_2)\equiv 1$, formulas \eqref{gdstsj6e5} become the canonical commutation relations (CCR), and for $Q(x_1,x_2)\equiv -1$ they become the canonical anticommutation relations (CAR).

\subsection{Quasi-free states}

In the theory of the CCR and CAR algebras, quasi-free states play a fundamental role. We refer the reader to e.g.\ \cite{A,AS}, Chapter~5 in \cite{BR}, or Chapter~17 in  \cite{DG}.

Let $\mathbb A$ be the  unital $*$-algebra generated by the operators $A^+(f)$, $A^-(f)$  ($f\in \mathcal H$), satisfying either the CCR or CAR. Define, for $f\in \mathcal H$, the field (or Segal-type) operators 
$B(f)=A^+(f)+A^-(Jf)$. Since 
$A^+(f)=\frac{1}{2}\big(B(f)-iB(if)\big)$ and $A^-(f)=\frac{1}{2}\big(B(Jf)+iB(iJf)\big)$, the algebra $\mathbb A$ is generated by $B(f)$ ($f\in \mathcal H$). Let $\tau$ be a state on $\mathbb A$. Then $\tau$ is completely determined by the moments of $B(f)$  under $\tau$. Intuitively, the state $\tau$ being quasi-free means that the moments of $B(f)$ can be calculated similarly to the case where $\tau$ is the vacuum state on the Fock representation of the CCR or CAR, respectively. For example, in the case of the CAR, this means that the odd moments are equal to zero and 
\begin{equation}\label{ye5y373}
\tau\big(B(f_1)B(f_2)\dotsm B(f_{2n})\big)=\sum_\xi (-1)^{\operatorname{cross}(\xi)}\prod_{\substack
{\{i,j\}\in \xi\\ i<j}} \tau\big(B(f_i)B(f_j)\big).\end{equation}
Here, the summation is over all partitions $\xi$ of the set $\{1,2,\dots,2n\}$ into $n$ two-point sets and  $\operatorname{cross}(\xi)$ denotes the number of crossings in $\xi$, i.e., the number of all choices of $\{i,j\},\,\{k,l\}\in \xi$ with $i<k<j<l$.

A state $\tau$ is called gauge-invariant if, for any $q\in \mathbb C$, $\vert q\vert=1$, $\tau$ remains invariant under the transformation 
$A^+(f)\mapsto qA^+(f)$, $A^-(f)\mapsto \overline qA^-(f)$, equivalently 
$$\tau\big(B(qf_1)\dotsm B(qf_n)\big)=\tau\big(B(f_1)\dotsm B(f_n)\big)\quad\forall n.$$

Due to the commutation relations,  any state $\tau$ is completely characterized by the  $n$-point functions, 
\begin{equation}\label{ts6uebbb}
\mathbf S^{(m,n)}(f_1,\dots,f_m,g_1,\dots,g_n)=\tau\big(A^+(f_1)\dotsm A^+(f_m)A^-(g_1)\dotsm A^-(g_n)\big).
\end{equation}
As easily seen, a state $\tau$ is gauge-invariant if and only if $\mathbf S^{(m,n)}=0$ if $m\neq n$. In fact, a state $\tau$ is gauge-invariant quasi-free if and only if $\mathbf  S^{(m,n)}=0$ if $m\neq n$ and 
\begin{equation}
\mathbf S^{(n,n)}(f_n,\dots,f_1,g_1,\dots,g_n)=\operatorname{per}\left[\mathbf S^{(1,1)}(f_i,g_j)\right]_{i,j=1,\dots,n}=\sum_{\pi \in S_n} \prod_{i=1}^n \mathbf  S^{(1,1)}(f_i,g_{\pi(i)})\label{cfs5w5u64wu}
\end{equation}
for the CCR algebra, and 
\begin{equation}\label{jfdu5wu5}
\mathbf S^{(n,n)}(f_n,\dots,f_1,g_1,\dots,g_n)=\operatorname{det}\left[\mathbf S^{(1,1)}(f_i,g_j) \right]_{i,j=1,\dots,n}=\sum_{\pi \in S_n}\text{sgn}(\pi) \prod_{i=1}^n \mathbf S^{(1,1)}(f_i,g_{\pi(i)})\end{equation}
for the CAR algebra. 

Araki, Woods \cite{AWoods} and Araki, Wyss \cite{AWyss} presented an explicit construction of the gauge-invariant quasi-free states on the CCR algebra and the CAR algebra, respectively. This has been done by constructing a non-Fock representation of the commutation relations in the symmetric, respectively antisymmetric Fock space over $\mathcal H\oplus\mathcal H$ (doubling of the underlying space) and applying the vacuum state to this representation.  

In \cite{anyons}, gauge-invariant quasi-free states on the algebra of the anyon commutation relations (ACR) were defined and constructed. The definition of such a state was based on an extension of formulas \eqref{cfs5w5u64wu}, \eqref{jfdu5wu5} to the ACR case. In particular, a  counterpart of the sign of the permutation $\pi$ in formula \eqref{jfdu5wu5} becomes the function 
\begin{equation}\label{vcytrd6u4we}
Q_{\pi}(x_1,\dots,x_n)=\prod_{\substack{1\le i<j\le n\\ \pi(i)>\pi(j)} }Q(x_i,x_j).\end{equation}
Note that, for $Q(\cdot,\cdot)\equiv -1$, one indeed  has $Q_\pi(x_1,\dots,x_n)\equiv\text{sgn}(\pi)$. A gauge-invar\-iant quasi-free state on the ACR algebra was constructed  as the vacuum state on a  representation of the ACR in the Fock space of $\mathbf Q$-symmetric functions with the underlying space $\mathbf X:=X_1\sqcup X_2$. Here $X_1$ and $X_2$ are two copies of the space $X$ and the function $\mathbf Q$ is defined as follows: $\mathbf Q(x_1,x_2):=Q(x_1,x_2)$ if both $x_1$ and $x_1$ are from the same part, $X_i$ ($i=1,2$), and $\mathbf Q(x_1,x_2):=Q(x_2,x_1)$ if $x_1$ and $x_2$ are from  different parts, $X_i$ and $X_j$, respectively, with $i\ne j$. In the case of a CCR or CAR algebra this construction of gauge-invariant quasi-free states  reduces to the construction of Araki--Woods or Araki--Wyss, respectively. 

\subsection{A brief description of the results}

In this paper, we consider a multicomponent system with a continuous operator-valued function $Q(x_1,x_2)$. We define a $Q$-MCR algebra, which contains multiple integrals as those appearing on the right-hand side of formulas \eqref{xseas5aq53q}--\eqref{yqdqidi}. 

For the vacuum state defined on the Fock representation of the $Q$-MCR, we calculate the moments  of the field operators. Similarly to \eqref{ye5y373}, this allows us to come up with formulas which may potentially be required for a state on the  $Q$-MCR algebra to be quasi-free. 
However, it appears that even  the gauge-invariant quasi-free states on the $Q$-ACR algebra constructed in  \cite{anyons} do not satisfy these conditions. The reason for this is the very definition of the function $\mathbf  Q(x_1,x_2)$. Indeed, when one evaluates the $n$-point functions $\mathbf  S^{(n,n)}$ for such a system, only the terms with all points from $X_1$ do not vanish. But for such points, $x_1,x_2\in X_1$, we have $\mathbf  Q(x_1,x_2)=Q(x_1,x_2)$. 
Hence, one comes up with an extension of formulas \eqref{cfs5w5u64wu}, \eqref{jfdu5wu5}  which contains the function $Q_\pi$ given by \eqref{vcytrd6u4we} in place of the sign of the permutation $\pi$ in the case of the CAR.  However, when one evaluates $\tau\big(B(f_1)\dotsm B(f_{2n})\big)$, one has to deal with terms containing points from both parts $X_1$ and $X_2$. Because of this, one is unable to come up with the function $Q_{\pi}(x_1,\dots,x_n)$, unless $Q(x_1,x_2)=Q(x_2,x_1)$. But the latter formula means that $Q(x_1,x_2)$ is either identically equal to $1$ (CCR), or identically equal to $-1$ (CAR). 

Thus, instead of using the term a {\it quasi-free state on the $Q$-MCR algebra}, we use the term a {\it strongly quasi-free state on the $Q$-MCR algebra} if a counterpart of formula \eqref{ye5y373} holds for  a state $\tau$ on the $Q$-MCR algebra. As a result, the gauge-invariant quasi-free states constructed on the $Q$-ACR algebra in \cite{anyons} are not strongly quasi-free. If $\mathbf S^{(m,n)}=0$ for $m\ne n$ and a counterpart of formulas \eqref{cfs5w5u64wu}, \eqref{jfdu5wu5} holds  for  a state $\tau$ on the $Q$-MCR algebra, then we call $\tau$ {\it gauge-invariant quasi-free}.

Our next aim is to construct gauge-invariant quasi-free states on the $Q$-MCR algebra. One hopes that the construction from \cite{anyons} of the gauge-invariant quasi-free states on the  $Q$-ACR algebra (which in turn includes the Araki--Woods and Araki--Wyss constructions as special cases) can be extended to the case of a multicomponent system. We show that this indeed can be achieved, however under  additional assumptions on the operator-valued function $Q$. It should be stressed that these assumptions are essentially  necessary for the $Q$-MCR to hold for this construction. 

Furthermore, if it additionally holds that $\widetilde Q(x_1,x_2)=Q(x_2,x_1)$, then the corresponding gauge-invariant quasi-free states are also strongly quasi-free.

We show that the developed theory is applicable to systems that contain quasiparticles of opposite type. An example of such a system is  a two-component system in which two quasiparticles, under exchange, change their respective types to the opposite ones  ($1\mapsto 2$, $2\mapsto1$).  More precisely, let $q_1$ and $q_2$ be complex-valued continuous functions on $\R^2$ satisfying $\overline{q_i(y_1,y_2)}=q_i(y_2,y_1)$, $|q_i(y_1,y_2)|=1$, and let the permutation $\theta\in S_2$ be given by $\theta(1)=2$, $\theta(2)=1$. Then the $Q$-MCR in this case are of the following form: 
\begin{align}
A_i^+(x_1)A_i^+(x_2)&=q_1(y_2,y_1)A_{\theta(i)}^+(x_2)A_{\theta(i)}^+(x_1),\notag\\
 A_i^+(x_1)A_{j}^+(x_2)&=q_2(y_2,y_1)A_i^+(x_2)A_{j}^+(x_1),\quad i\ne j,\notag\\
A_i^-(x_1)A_i^-(x_2)&=q_1(y_2,y_1)A_{\theta(i)}^-(x_2)A_{\theta(i)}^-(x_1),\notag\\
 A_i^-(x_1)A_{j}^-(x_2)&=q_2(y_2,y_1)A_i^-(x_2)A_{j}^-(x_1),\quad i\ne j,\notag\\
A_i^-(x_1)A_i^+(x_2)&=\delta(x_1-x_2)+q_2(y_1,y_2)A_{\theta(i)}^+(x_2)A_{\theta(i)}^-(x_1),\notag\\
 A_i^-(x_1)A_{j}^+(x_2)&=q_1(y_1,y_2)A_i^+(y_2)A_{j}^-(x_1),\quad i\ne j,\label{utfdT}
\end{align}
for $i,j\in\{1,2\}$.    For the corresponding $Q$-MCR algebra $\mathbb A$, our theory gives a class of gauge-invariant quasi-free states $\tau$ on $\mathbb A$.

If we additionally assume that $q_1(y_1,y_2)=q_2(y_2,y_1)$, then   the condition $\widetilde Q(x_1,x_2)=Q(x_2,x_1)$ is satisfied. Hence,  in this case, each gauge-invariant quasi-free state $\tau$ constructed in the paper is  strongly quasi-free, a property that the gauge-invariant quasi-free states on the ACR algebra do not possess. 

Fusion of (non-abelian) anyons plays a central role in topological quantum computation, see e.g.\ \cite[Chapter~4]{Pachos} or \cite[Section~12.1]{Stanescu}.  Intuitively, fusion means putting several anyons in an infinitely small box and identifying the statistical behaviour of the box. Let $k\ge3$ be an odd number. We show that fusion of $k$ quasiparticles that are  described by the commutation relations  \eqref{utfdT} leads to a new, nontrivial exchange function $Q(x_1,x_2)$ to which our theory of quasi-free states is applicable. Note that, in this case, we have $V=(\mathbb C^2)^{\otimes k}$, i.e., the quantum system has  $2^k$ components. 

The paper is organized as follows. In Section~\ref{vcyrte6ie4w}, we define the $Q$-MCR algebra and discuss some of its basic properties. In particular, we show that the $Q$-MCR algebra allows Wick (normal) ordering. In Section~\ref{d6es6use}, we briefly discuss general states on the $Q$-MCR algebra. In Section~\ref{fgdyjdde}, we construct the Fock state  on the $Q$-MCR algebra. The main result of this section, Theorem~\ref{radiaition12345}, gives an explicit formula for the moments of sums of creation and annihilations operators with respect to the Fock state.  

The main results of the paper are in Section~\ref{kljiun6958} and  \ref{vcrte6u}. In Section~\ref{kljiun6958}, we use Theorem~\ref{radiaition12345} to define gauge-invariant quasi-free states and strongly quasi-free states on the $Q$-MCR algebra. Under several assumptions on the exchange (operator-valued) function $Q$, we explicitly construct such quasi-free states (Theorem~\ref{buty7e6u}). In Section~\ref{vcrte6u}, we present examples of the function $Q$ that satisfies the conditions of  Theorem~\ref{buty7e6u}. In particular, Theorem~\ref{rtsw5uw5ude} states the form of a function $Q$ that satisfies the conditions of  Theorem~\ref{buty7e6u} and describes a quantum system with quasiparticles of opposite type.  

Finally, we remark that, in the case where the condition $\widetilde Q(x_1,x_2)=Q(x_2,x_1)$ is satisfied and hence strongly quasi-free states exist, one may hope to construct strongly quasi-free states that are not gauge-invariant.  

\section{The $Q$-MCR algebra}\label{vcyrte6ie4w}

Let $Q:X^2\to\mathcal L(V^{\otimes 2})$ be a continuous function as in Subsection~\ref{cydyccdfe}. Our aim is to define the $Q$-MCR algebra. 

To simplify notation, we denote $Y:=\R$, $Z:=\R^{d-1}$, so that $X=Y\times Z$. Respectively, for a given $x\in X$, we denote $y:=x^1\in Y$ and $z:=(x^2,\dots,x^d)\in Z$, so that $x=(y,z)$. We denote $\mathcal G:=L^2(Z;\mathbb C)$. Then   
$$\mathcal H=L^2(X; V)=L^2(Y; V)\otimes \mathcal G.$$
 We denote by  $C_0(Y^n;V^{\otimes n}) $ the vector space of all continuous functions $\varphi^{(n)}:Y^n \to V^{\otimes n}$ with compact support. We denote by 
 $\mathfrak F^{(n)}$ the vector space spanned  by functions $f^{(n)}:X^n\to V^{\otimes n}$  of the form 
\begin{equation}\label{tsar43}
f^{(n)}(x_1,\dots,x_n)=\varphi^{(n)}(y_1,\dots,y_n)g_1(z_1),\dotsm g_n(z_n)\end{equation}
 with $\varphi^{(n)} \in C_0(Y^n ;V^{\otimes n})$ and $g_1,\dots,g_n \in \mathcal G$. 
 
 \begin{remark} Below we will use the trivial facts  that, for $f^{(m)}\in \mathfrak F^{(m)}$,   $g^{(n)}\in \mathfrak F^{(n)}$ and  a continuous function $C:Y^2\to \mathcal L(V^{\otimes 2})$,
 $$f^{(m)}(x_1,\dots,x_m)\otimes g^{(n)}(x_{m+1},\dots,x_{m+n})\in  \mathfrak F^{(m+n)},$$
and 
 $$ C_i(y_i,y_{i+1})f^{(m)}(x_1,\dots,x_{i+1},x_i,\dots,x_m) \in  \mathfrak F^{(m)},\quad i=1,\dots,m-1.$$
Furthermore, for $u,v\in V$, $\langle u,v\rangle_V=\operatorname{Tr}(u\otimes v)$.
 \end{remark}

We define the $Q$-MCR algebra $\mathbb A$ as follows. Consider the following (formal for now) operator-valued integrals:
\begin{equation}\label{vtrs56e57}
\int_{X^n} \langle f^{(n)}(x_1,\dots,x_n), A^{\sharp_1}(x_1)\otimes \dotsm \otimes A^{\sharp_n}(x_n)\rangle_{V^{\otimes n}}\, dx_1\dotsm  dx_n \end{equation}
with $f^{(n)}\in \mathfrak F^{(n)}$ and $\sharp_1,\dots,\sharp_n \in \{+,-\}$ ($n\in\mathbb N$), the integral in \eqref{vtrs56e57} depending on $f^{(n)}$ linearly. To shorten notation, we will denote the integral in \eqref{vtrs56e57} by
\begin{equation}\label{cftrs5y64e}
\Phi(f^{(n)};\sharp_1,\dots,\sharp_n)=\Phi(f^{(n)}(x_1,\dots,x_n);\sharp_1,\dots,\sharp_n).\end{equation}
As a vector space, $\mathbb A$ is spanned by the identity operator and integrals of the form \eqref{vtrs56e57}, equivalently \eqref{cftrs5y64e}.
Then $\mathbb A$  is defined as the unital $*$-algebra with the product 
\begin{align}
&\Phi(f^{(m)};\sharp_1,\dots,\sharp_m)\Phi(g^{(n)};\sharp_{m+1},\dots,\sharp_{m+n})\notag\\
&\quad =\Phi(f^{(m)}(x_1,\dots,x_m)\otimes g^{(n)}(x_{m+1},\dots,x_{m+n});\sharp_1,\dots,\sharp_{m+n})\label{tera4qy}
\end{align}
and the $*$-operation 
\begin{equation}\label{vftrd6e7}
\Phi(f^{(n)};\sharp_1,\dots,\sharp_n)^*= \Phi(\mathbb S^{(n)}f^{(n)}(x_n,\dots,x_1);-\sharp_n,\dots,-\sharp_1).
\end{equation}
Here, for $\sharp\in\{+,-\}$, $-\sharp$ denotes the opposite sign of $\sharp$,   and 
 the antilinear operator $\mathbb S^{(n)}: V^{\otimes n}\to V^{\otimes n}$  is defined by 
\begin{equation*}
\mathbb S^{(n)}v_1\otimes \dotsm \otimes v_n=(Jv_n)\otimes \dotsm \otimes (Jv_1),\quad v_1,\dots,v_n\in V.
\end{equation*}

Furthermore, elements of $\mathbb A$ are subject to the commutation relations \eqref{xseas5aq53q}--\eqref{yqdqidi}. More precisely, if $i\in \{1,\dots,n-1\}$ and $\sharp_i=\sharp_{i+1}=+$, then 
\begin{equation}\label{gdtrs5u4}
\Phi(f^{(n)};\sharp_1,\dots,\sharp_n)=\Phi\big(Q_i(x_i,x_{i+1})f^{(n)}(x_1,\dots,x_{i+1},x_i,\dots,x_n);\sharp_1,\dots,\sharp_n\big),
\end{equation}
if $\sharp_i=\sharp_{i+1}=-$, then 
\begin{equation}\label{vcerw5q5q}
\Phi(f^{(n)};\sharp_1,\dots,\sharp_n)=\Phi\big(\widehat Q_i(x_{i},x_{i+1})f^{(n)}(x_1,\dots,x_{i+1},x_i,\dots,x_n);\sharp_1,\dots,\sharp_n\big),
\end{equation}
and if  $\sharp_i=-,\,\sharp_{i+1}=+$, then
\begin{align}
\Phi(f^{(n)};\sharp_1,\dots,\sharp_n)&=\Phi(\widetilde Q_i(x_{i+1},x_i)f^{(n)}(x_1,\dots,x_{i+1},x_i,\dots, x_n);\sharp_1,\dots,\sharp_{i+1},\sharp_i,\dots,\sharp_n)\notag\\
&\quad + \Phi(g^{(n-2)};\sharp_1,\dots,\sharp_{i-1},\sharp_{i+2},\dots,\sharp_n),\label{ctrw53}
\end{align}
where
\begin{equation}\label{cdtesw5w}
g^{(n-2)}(x_1,\dots,x_{n-2}):= \int_X \operatorname{Tr}_if^{(n)}(x_1,\dots,x_{i-1},x,x,x_{i},\dots,x_{n-2})\,dx\in \mathfrak F^{(n-2)}. \end{equation}
Here  $\operatorname{Tr}_i:V^{\otimes n}\to V^{\otimes (n-2)}$ acts by the functional $\operatorname{Tr}:V^{\otimes 2} \to\mathbb C$ on the $i$th and $(i+1)$th components of $V^{\otimes n}$. 

For $m,n\ge0$ with $m+n\ge1$ and $f^{(m+n)}\in \mathfrak F^{(m+n)}$, denote $W^{(m,n)}(f^{(m+n)}):=\Phi(f^{(m+n)};\sharp_1,\dots,\sharp_{m+n})$
with $\sharp_1=\dots=\sharp_m=+$, $\sharp_{m+1}=\dots=\sharp_{m+n}=-$. Thus, $W^{(m,n)}(f^{(m+n)})$ is a Wick-ordered element of the $Q$-MCR algebra $\mathbb A$.
Let also $\mathfrak{SF}^{(m,n)}$ consist of all elements $f^{(m+n)}\in \mathfrak{F}^{(m+n)}$ that satisfy 
\begin{equation}\label{vcfydst6}
f^{(m+n)}(x_1,\dots,x_{m+n})=Q_i(x_i,x_{i+1})f^{(m+n)}(x_1,\dots,x_{i+1},x_i,\dots,x_{m+n})
\end{equation}
for all $i =1,\dots, m-1$,  and 
\begin{align*}
f^{(m+n)}(x_1,\dots,x_{m+n})=\widehat Q_i(x_{i},x_{i+1})f^{(m+n)}(x_1,\dots,x_{i+1},x_i,\dots,x_{m+n})
\end{align*}
for all $i=m+1,\dots,m+n-1$. In words, $f^{(m+n)}$ is $Q$-symmetric in the first $m$ variables and $\widehat Q$-symmetric in the last $n$ variables. 

The following proposition states that $\mathbb A$ is a Wick algebra, i.e., it allows Wick (normal) ordering. 

\begin{proposition}\label{tye7i43}
Each element of the $Q$-MCR algebra $\mathbb A$ can be represented in the form 
\begin{equation}\label{askjndsx}
c\mathbf{1}+\sum_{m,n\ge0,\ m+n \ge 1}W^{(m,n)}(f^{(m,n)}) 
\end{equation}
where $c \in \mathbb C$ and $f^{(m,n)} \in \mathfrak {SF}^{(m,n)}$. The sum in \eqref{askjndsx} is finite.
\end{proposition}

\begin{proof} The commutation relation \eqref{ctrw53} implies that each element of $\mathbb A$ can be represented in the form \eqref{askjndsx} with $c\in\mathbb C$ and $f^{(m,n)}\in\mathfrak F^{(m+n)}$.  So we only need to prove that the functions $f^{(m,n)}$ can be chosen from $\mathfrak {SF}^{(m,n)}$. 

For $i=1,\dots,m-1$, define $U_i:\mathfrak F^{(m+n)}\to \mathfrak F^{(m+n)}$ by 
\begin{equation}\label{vcrtw46ue3i6}
(U_if^{(m+n)})(x_1,\dots,x_{m+n}):=Q_i(x_i,x_{i+1})f^{(m+n)}(x_1,\dots,x_{i+1},x_i,\dots,x_{m+n}).\end{equation}
The functional Yang--Baxter equation  \eqref{v6w6u34} implies that, for $i=1,\dots,m-2$,  $U_iU_{i+1}U_i=U_{i+1}U_iU_{i+1}$ (the Yang--Baxter equation). Furthermore, we obviously have $U_iU_j=U_jU_i$ if $|i-j|\ge2$, and $U_i^2=\mathbf 1$. Therefore, we can construct a representation of the symmetric group $S_m$ by setting $U_\pi:=U_{i_1}U_{i_2}\dotsm U_{i_k}$ where $\pi \in S_m$ has an (arbitrary) representation $\pi=\pi_{i_1}\pi_{i_2}\dotsm \pi_{i_k}$. Here, for $i=1,\dots,m-1$, $\pi_i$ is the adjacent transposition of $i$ and $i+1$. Let $P_m:=\frac1{m!}\sum_{\pi\in S_m}U_\pi$. 

Next, for $i=m+1,\dots,m+n-1$, we define  $\widehat U_i:\mathfrak F^{(m+n)}\to \mathfrak F^{(m+n)}$ by 
$$(\widehat U_if^{(m+n)})(x_1,\dots,x_{m+n}):=\widehat Q_i(x_{i},x_{i+1})f^{(m+n)}(x_1,\dots,x_{i+1},x_i,\dots,x_{m+n}).$$
It follows from \eqref{v6w6u34} and the definition of $\widehat Q(\cdot,\cdot)$ that  $\widehat Q(\cdot,\cdot)$  also satisfies the functional Yang--Baxter equation. Hence, we can similarly define $\widehat P_n:=\frac1{n!}\sum_{\pi\in S_n}\widehat U_\pi$, where $S_n$ is interpreted as the group of permutations of $m+1,\dots,m+n$. 

It is easy to see that $P_m\widehat P_n$ maps $\mathfrak F^{(m+n)}$ onto $\mathfrak {SF}^{(m,n)}$. Furthermore, by  \eqref{gdtrs5u4} and \eqref{vcerw5q5q}, 
for each $f^{(m+n)}\in\mathfrak F^{(m+n)}$, we have 
$$W^{(m,n)}(f^{(m+n)})=W^{(m,n)}(P_m\widehat P_n f^{(m+n)}).\qquad\qedhere $$
\end{proof}

The following proposition shows that, under an additional assumption on the operator-valued function $Q$, the algebra $\mathbb A$ allows anti-normal ordering.

\begin{proposition}\label{zzzzaaqqq}
Assume that, for all $(x_1,x_2)\in X^2$, the operator $\widetilde Q(x_1,x_2)\in\mathcal L(V^{\otimes 2}) $ is invertible. Furthermore, assume that there exists a constant $\varkappa \in \mathbb R$ such that, for all $x \in X$ and $v^{(2)} \in V^{\otimes 2}$, 
\begin{equation}\label{olhkgmknenj}
\operatorname{Tr}\big(\widetilde Q(x,x)^{-1}v^{(2)}\big)=\varkappa \operatorname{Tr}v^{(2)},\quad x\in X,\ v^{(2)} \in V^{\otimes 2}.
\end{equation}
Let $\sharp_1,\dots,\sharp_n\in\{+,-\}$ be such that, for some $i\in\{1,\dots,n-1\}$, $\sharp_i=+$ and $\sharp_{i+1}=-$. Then, for all $f^{(n)}\in\mathfrak F^{(n)}$,
\begin{align}
&\Phi^{(n)}\left(f^{(n)};\sharp_1,\dots,\sharp_n\right)\notag\\
&\quad=\Phi^{(n)}\left(\widetilde Q_i(x_i,x_{i+1})^{-1}f^{(n)}(x_1,\dots,x_{i+1},x_i,\dots,x_n); \sharp_1,\dots,\sharp_{i+1},\sharp_i,\dots,\sharp_n\right)\notag\\
&\qquad-\varkappa  \Phi(g^{(n-2)};\sharp_1,\dots,\sharp_{i-1},\sharp_{i+2},\dots,\sharp_n),\label{xaw45w}
\end{align}
where  $g^{(n-2)}$ is given by \eqref{cdtesw5w}.
\end{proposition}

\begin{proof}
Applying formulas \eqref{ctrw53}, \eqref{cdtesw5w} to 
$$\Phi^{(n)}\left(\widetilde Q_i(x_i,x_{i+1})^{-1}f^{(n)}(x_1,\dots,x_{i+1},x_i,\dots,x_n); \sharp_1,\dots,\sharp_{i+1},\sharp_i,\dots,\sharp_n\right)$$ 
and using the assumption \eqref{olhkgmknenj}, we get  \eqref{xaw45w}. \end{proof}

\section{States on the $Q$-MCR algebra}\label{d6es6use}

To study states on the $Q$-MCR algebra $\mathbb A$,  we need to make an assumption on their continuity.
To this end, we equip $C_0(Y^n;V^{\otimes n})$ with the Fr\'echet topology in which a sequence $(f_k)_{k=1}^\infty$ converges if and only if there exists a compact set $K\subset Y^n$ such that the support of each $f_k$ is a subset of $K$ and $(f_k)_{k=1}^\infty$ converges uniformly on $K$.
 The dual space of $C_0(Y^n; V^{\otimes n})$ is the space $\mathcal M(Y^n; V^{\otimes n})$ of $V^{\otimes n}$-valued Radon measures on $Y^n$, i.e., $F^{(n)}$ is a continuous linear functional  on $C_0(Y^n; V^{\otimes n})$ if and only if 
$$F^{(n)}(\varphi^{(n)})=\int_{Y^n}\langle \varphi^{(n)}(y_1,\dots,y_n),m^{(n)}(dy_1\dots, dy_n)\rangle_{V^{\otimes n}},\quad \varphi^{(n)}\in C_0(Y^n; V^{\otimes n}),$$
 where $m^{(n)}\in \mathcal M(Y^n; V^{\otimes n})$, see e.g.\ \cite[Chapter~7]{Simonnet}. Note that $m^{(n)}$ is completely determined by its values on functions of the form $\varphi^{(n)}(y_1,\dots,y_n)=\varphi_1(y_1)\otimes\dots\otimes  \varphi_n(y_n)$ with $\varphi_1,\dots,\varphi_n\in C_0(Y; V)$. 

Let  $\tau:\mathbb A \to  \mathbb C$ be a state on $\mathbb A$. For any $\sharp_1,\dots,\sharp_n\in \{+,-\}$,  we define a linear functional 
$\tau^{(n)}_{\sharp_1,\dots,\sharp_n}:C_0(Y^n; V^{\otimes n})\times \mathcal G^n \to \mathbb C$ by 
\begin{align}
\tau^{(n)}_{\sharp_1,\dots,\sharp_n }(\varphi^{(n)},g_1,\dots,g_n):=\tau \left(\Phi^{(n)}(\varphi^{(n)}(y_1,\dots,y_n)g_1(z_1)\dotsm g_n(z_n);\sharp_1,\dots,\sharp_n)\right).\label{dre5te5}
\end{align}
Obviously, these   linear functionals uniquely identify the state $\tau$. We assume that, for any fixed $g_1,\dots,g_n\in\mathcal G$, $\tau^{(n)}_{\sharp_1,\dots,\sharp_n }(\cdot,g_1,\dots,g_n)$ is continuous on $C_0(Y^n; V^{\otimes n})$.
Hence, there exists a $V^{\otimes n}$-valued Radon measure $m_{\sharp_1,\dots,\sharp_n}^{(n)}[g_1,\dots,g_n]\in  \mathcal M(Y^n; V^{\otimes n})$  that can be identified with $\tau^{(n)}_{\sharp_1,\dots,\sharp_n }(\cdot,g_1,\dots,g_n)$. 

Below, for $f\in\mathfrak F^{(1)}$ and $\sharp\in\{+,-\}$, we denote $A^\sharp(f):=\Phi(f;\sharp)$, called creation, respectively  annihilation operators in the $Q$-MCR algebra $\mathbb A$, and 
\begin{equation}\label{buyyr7}
B(f):=A^+(f)+A^-(Jf). 
\end{equation}

\begin{proposition}\label{vcreaq4y5qwy} Under the above assumption on the continuity of a state $\tau$ on a $Q$-MCR algebra $\mathbb A$, the following statements hold.

(i) For $m,n\ge0$, $m+n\ge1$, we define 
\begin{align}
&\mathbf S^{(m,n)}(\varphi_1\otimes g_1,\dots,\varphi_m\otimes g_m,\varphi_{m+1}\otimes g_{m+1},\dots,\varphi_{m+n}\otimes g_{m+n})\notag\\
&\quad:= \tau\big(A^+(\varphi_1\otimes g_1)\dotsm A^+(\varphi_m\otimes g_m)A^- (\varphi_{m+1}\otimes g_{m+1})\dotsm A^-(\varphi_{m+n}\otimes g_{m+n})\big), \notag                  
\end{align}
where $\varphi_i\in C_0(Y;V)$, $g_i\in\mathcal G$, and $(\varphi_i\otimes g_i)(x):=\varphi_i(y)g_i(z)$ ($i=1,\dots,m+n$).
Then the functionals  $\mathbf S^{(m,n)}$ uniquely identify the state $\tau$.

(ii) For $n\ge 1$, we define
\begin{equation*}
\mathbf M^{(n)}(\varphi_1\otimes g_1,\dots,\varphi_n\otimes g_n):=\tau\big(B(\varphi_1\otimes g_1)\dotsm B(\varphi_n\otimes g_n)\big),\end{equation*}
where  $\varphi_i\in C_0(Y; V_\R)$,  $g_i\in\mathcal G$ ($i=1,\dots,n$). Here $V_\mathbb R:=\R^m\subset\mathbb C^m=V$. Then the functionals $\mathbf M^{(n)}$ uniquely identify the state $\tau$.

\end{proposition}

\begin{proof} (i) By Proposition~\ref{tye7i43}, the state $\tau$ is completely determined by the functionals $\tau^{(m+n)}_{\sharp_1,\dots,\sharp_{m+n}}$ with $\sharp_1=\dots=\sharp_m=+$ and $\sharp_{m+1}=\dots=\sharp_{m+n}=-$, where $m,n\ge0$, $m+n\ge1$. Statement (i) now follows from the assumed continuity of the functionals $\tau^{(m+n)}_{\sharp_1,\dots,\sharp_{m+n}}$. 

(ii) First, we note that the result of part (i) remains true if all $\varphi_i$'s are assumed to be from $C_0(Y; V_{\mathbb R})$. Next, for each $\varphi\in C_0(Y; V_{\mathbb R})$ and $g\in\mathcal G$, we have
$$
A^+(\varphi\otimes g)=\frac12\big(B(\varphi\otimes g)-iB(\varphi\otimes ig)\big),\quad A^-(\varphi\otimes g)=\frac12\big(B(\varphi\otimes \bar g)
+iB(\varphi\otimes i\bar g)\big).
$$
Hence, the functionals $\mathbf M^{(n)}$ uniquely identify the functionals $\mathbf S^{(m,n)}$.
\end{proof}

\section{The Fock state on the $Q$-MCR algebra}\label{fgdyjdde}
 
 Let us briefly recall the Fock representation of the $Q$-MCR \eqref{xseas5aq53q}--\eqref{yqdqidi}, for detail see \cite{LM} or \cite{DKLP}. First, note that $\mathcal H^{\otimes n}=L^2(X^n; V^{\otimes n})$. Define a unitary operator $U\in\mathcal L(\mathcal H^{\otimes 2})$ by $(Uf^{(2)})(x_1,x_2):=Q(x_1,x_2)f^{(2)}(x_2,x_1)$. For $i=1,\dots,n-1$, let $U_i$ denote the operator $U$ acting on the $i$th and $(i+1)$th components of $\mathcal H^{\otimes n}$, compare with \eqref{vcrtw46ue3i6}. Similarly to the proof of Proposition~\ref{tye7i43}, define $P_n:=\sum_{\pi\in S_n}U_\pi$. Then $P_n$ is an orthogonal projection in $\mathcal H^{\otimes n}$. Denote by $\mathcal H^{\circledast n}$ the image of $P_n$. This is the subspace of all functions from $\mathcal H^{\otimes n}$ that are $Q$-symmetric, compare with \eqref{vcfydst6}. Let also $\mathcal H^{\circledast 0}:=\mathbb C$. Define the $Q$-symmetric Fock space $\mathcal F^Q(\mathcal H):=\bigoplus_{n=0}^\infty \mathcal H^{\circledast n} n!$, i.e., $\mathcal F^Q(\mathcal H)$ is the Hilbert space of all infinite sequences $F=(f^{(n)})_{n=0}^\infty$ with $f^{(n)}\in \mathcal H^{\circledast n} $ and $\|F\|^2_{\mathcal F^Q(\mathcal H)}=\sum_{n=0}^\infty \|f^{(n)}\|_{\mathcal H^{\circledast n} }^2n!<\infty$. The vector $\Omega=(1,0,0,\dots)$ is called the vacuum.  Let $\mathcal F^Q_{\mathrm{fin}}(\mathcal H)$ denote the (dense) subspace of $\mathcal F^Q(\mathcal H)$ consisting of all $F=(f^{(n)})_{n=0}^\infty\in \mathcal F^Q(\mathcal H)$ such that, for some $N\in\mathbb N$ (depending on $F$), $f^{(n)}=0$ for all $n\ge N$. 
 
Let $f\in\mathcal H$. We define a (standard) creation operator $a^+(f)$ as the linear operator in $\mathcal                    F^Q_{\mathrm{fin}}(\mathcal H)$ given by $a^+(f)f^{(n)}:=P_{n+1}(f\otimes f^{(n)})$ for each $f^{(n)}\in\mathcal H^{\circledast n}$.  
  Next,  we define a  (standard) annihilation operator $a^-(f)$ as the linear operator in $\mathcal F^Q_{\mathrm{fin}}(\mathcal H)$ given by 
  \begin{equation}\label{cxtes5uw}
  (a^-(f)f^{(n)})(x_1,\dots,x_{n-1})=n\int_X\langle f^{(n)}(x,x_1,\dots,x_{n-1}),f(x)\rangle_V\,dx.\end{equation}  
  Here, for $u_1,\dots,u_n,v\in V$, we denote 
\begin{equation}\label{vcftse5y}
\langle u_1\otimes\dots\otimes u_n,v\rangle_V:=\langle u_1,v\rangle_V\,u_2\otimes\dots\otimes u_n\,.\end{equation} 
 Then $a^-(f)$ is the restriction to $\mathcal F^Q_{\mathrm{fin}}(\mathcal H)$ of $a^+(Jf)^*$.  
  
  Similarly to \eqref{vftye6u}, \eqref{vctesay5ra5} one defines operator-valued distributions $a^+(x)$, $a^-(x)$.  One  shows that the operator-valued integrals on the right-hand side of formulas \eqref{xseas5aq53q}--\eqref{yqdqidi}, with $A^+(\cdot)$, $A^-(\cdot)$ being replaced by $a^+(\cdot)$, $a^-(\cdot)$, are well-defined and formulas \eqref{xseas5aq53q}--\eqref{yqdqidi} hold. 
 
 Let us show that this construction can be easily extended to a representation of the $Q$-MCA algebra. 
By formula (33) in \cite{DKLP}, we have, for any $f_1,\dots,f_m\in\mathcal H$ and $f^{(k)}\in\mathcal H^{\circledast n}$,
\begin{equation}\label{ctesw5uy3w}
a^+(f_1)a^+(f_1)\dotsm a^+(f_m)f^{(k)}=P_{m+k}(f_1\otimes f_2\otimes\dotsm\otimes f_m\otimes f^{(k)}).\end{equation}
Furthermore, by using \eqref{cxtes5uw} and  \eqref{vcftse5y}, we have 
\begin{align}
&(a^-(f_1)a^-(f_2)\dotsm a^-(f_n)f^{(k)})(x_1,\dots,x_{k-n})=(k)_n \notag\\
&\times \int_{X^n}(\mathbb T^{(2n)}\otimes \mathbf 1_{k-n})(f_1\otimes\dots\otimes f_n)(x_1',\dots,x_n')\otimes f^{(k)}(x_n',x_{n-1}',\dots,x_1',x_1,\dots,x_{k-n})\notag\\
&\quad\times  dx_1'\dotsm dx_n'\,.\label{vtyqd6w}
 \end{align}
Here, $(k)_n$ is the Pochhammer symbol, $\mathbb T^{(2n)}:V^{\otimes(2n)}\to\mathbb C$ is the linear functional given by
\begin{equation}\label{qdfrqd}
\mathbb T^{(2n)}\,v_1\otimes\dots\otimes v_{2n}:=\prod_{i=1}^n\langle v_i,v_{2n-i+1}\rangle_V,\quad v_1,\dots,v_{2n}\in V, \end{equation}
and $\mathbf 1_{k-n}$ is the identity operator in $V^{\otimes(k-n)}$.

In view of formulas \eqref{ctesw5uy3w}, \eqref{vtyqd6w}, we define for $f^{(m+n)}\in\mathfrak F^{(m+n)}$, a linear operator  $W^{(m,n)}(f^{(m+n)})$ in $\mathcal F^Q_{\mathrm{fin}}(\mathcal H)$ by
\begin{align}
&\left(W^{(m,n)}\left(f^{(m+n)}\right)u^{(k)}\right)(x_1,\dots,x_{m-n+k})=P_{m-n+k}\bigg[\int_{X^n} \left(\mathbf{1}_m \otimes \mathbb T^{(2n)} \otimes \mathbf{1}_{k-n}\right)\notag\\
&f^{(m+n)}(x_1,\dots,x_m,x_1',\dots,x_n' )\otimes u^{(k)}(x_n',\dots,x_1',x_{m+1},\dots,x_{m-n+k})\,dx_1' \dotsm dx_n' \bigg]\label{gg44222}
\end{align}
for $u^{(k)}\in\mathcal H^{\circledast k}$. Due to the $Q$-MCR, the operators $W^{(m,n)}(f^{m+n})$ uniquely identify operators $\Phi(f^{(n)};\sharp_1,\dots,\sharp_n)$ with $f^{(n)}\in\mathfrak F^{(n)}$ and $\sharp_1,\dots,\sharp_n\in\{+,-\}$. Thus, we get a representation of the $Q$-MCR $*$-algebra $\mathbb A$.

Next, we define the vacuum state $\tau$ on $\mathbb A$ by $\tau(a):=(a\Omega,\Omega)_{\mathcal F^Q(\mathcal H)}$ for $a\in\mathbb A$. Since, for each $f\in\mathfrak F^{(1)}$, $a^-(f)\Omega=0$, $\tau$ is the Fock state on $\mathbb A$.
Note also that the state $\tau$ can be  extended to the vector space of all linear operators acting in  
$\mathcal F^Q_{\mathrm{fin}}(\mathcal H)$. 

Obviously, for the Fock state $\tau$, the corresponding functionals $\tau^{(n)}_{\sharp_1,\dots,\sharp_n }(\cdot,g_1,\dots,g_n)$ are continuous on $C_0(Y^n; V^{\otimes n})$. Hence, by Proposition~\ref{vcreaq4y5qwy}, the functionals $\mathbf S^{(m,n)}$ uniquely identify $\tau$, and so do the functionals $\mathbf M^{(n)}$. It follows from the definition of $\tau$ that all the functionals $\mathbf S^{(m,n)}$  are equal to zero.
Our next aim is to calculate the functionals $\mathbf M^{(n)}$. To this end, we will first prove  a slightly more general result. Before doing that,  let us introduce required notations.

For $n\in\mathbb N$, let $\mathcal P^{(2n)}$ denote the set of all partitions of $\{1,\dots,2n\}$ into $n$ two-point sets. Fix a partition $\xi \in \mathcal P^{(2n)}$. 
All the definitions below depend on the choice of $\xi$, nevertheless for simplicity we will not show this dependence in  some of our notations.

Let 
\begin{equation}\label{ytds}
\xi=\big\{\{i_1,j_1\},\dots, \{i_n,j_n\}\big\}\end{equation} with 
\begin{equation}\label{xzseRTYUI}
i_1<j_1,\ i_2<j_2,\dots,\ i_n <j_n, \quad i_1>i_2>\dotsm>i_n=1.
\end{equation} 
 Denote  $I:=\{i_1,i_2,\dots,i_n\}$, $J:=\{j_1,j_2,\dots,j_n\}$. For $k=1,2,\dots,n$, we define sets $J^{(k)}$ and $\mathbf J^{(k)}$ as follows.
Let 
$$
J^{(1)}:=\{j \in J \mid i_1<j \leq j_1\},\quad \mathbf J ^{(1)}:=\{j \in J \mid i_1<j\},$$
and for $k=2,\dots,n$, 
\begin{align}
J^{(k)}&:=\big\{j\in J \mid i_k <j\le j_k,\ j\neq j_1,\ j \neq j_2,\dots,\ j \neq j_{k-1}\big\},
\notag\\
\mathbf J^{(k)}&:=\left\{j\in J \mid i_k < j,\ j \neq \min \mathbf J^{(1)},\ j \neq \min \mathbf J^{(2)},\dots,\ j \neq \min \mathbf J^{(k-1)}\right\}.\label{xdzazza}
\end{align}
We write 
\begin{equation}\label{bgtd6w6u5w6u}
J^{(k)}=\big\{j_1^{(k)},j_2^{(k)},\dots,j_{l_k}^{(k)}\big\},\quad \mathbf J^{(k)}=\big\{\mathbf j_1^{(k)},\mathbf j_2^{(k)},\dots,\mathbf j_{m_k}^{(k)}\big\},
\end{equation}
with 
\begin{equation}\label{nifr75fdrd}
j_1^{(k)}<j_2^{(k)}<\dots<j_{l_k}^{(k)},\quad \mathbf j_1^{(k)}<\mathbf j_2^{(k)}<\dots<\mathbf j_{m_k}^{(k)}.
\end{equation}
Here  $l_k$ and $m_k$ are the number of elements of the sets $J^{(k)}$ and $\mathbf J^{(k)}$, respectively,


\begin{remark}\label{rds6qed}
Note that
\begin{align*}
\xi'=\big\{\{i_1,\mathbf j_1^{(1)}\},\{i_2,\mathbf j_1^{(2)}\},\dots,\{i_n,\mathbf j_1^{(n)}\}\big\}
\end{align*}
belongs to $\mathcal P^{(2n)}$. Furthermore, since  $\mathbf j_1^{(k)}=\min \mathbf J^{(k)}$, the definition \eqref{xdzazza} implies that $\xi'$ is a non-crossing partition.
\end{remark}

Recall that, for a linear operator $\mathcal C\in \mathcal L(V^{\otimes 2})$ and $i \in \{1,2,\dots,m-1\}$, we 
denote by $\mathcal C_i$ the linear operator in $V^{\otimes m}$ acting as $\mathcal C$ on the $i$th  and $(i+1)$th components of the tensor product $V^{\otimes m}$. Now, similarly, for $1\leq i <j \leq m$, we define a linear operator $ \mathcal C[i,j]$ in $V^{\otimes m}$ that acts as the operator $\mathcal C$ on the $i$th and $j$th components of  $V^{\otimes m}$.

Let us now fix an arbitrary $(x_1,\dots,x_{2n})\in X^{2n}$. For $k=1,\dots,n$, we define a linear operator $Q^{(k)}(\xi;x_1,x_2,\dots,x_{2n})$ in $V^{\otimes 2n}$ as follows. If $J^{(k)}=\{j_k\}$, i.e., $l_k=1$, then $Q^{(k)} (\xi;x_1,x_2,\dots,x_{2n})$ is the identity operator. 
If $l_k>1$, then we set 
\begin{align}
&Q^{(k)} (\xi;x_1,x_2,\dots,x_{2n})\notag
\\
&\quad:=Q(x_{i_k},x_{j_1^{(k)}})\big[\mathbf j_1^{(k)},\,\mathbf j_2^{(k)}\big]\, Q(x_{i_k},x_{j_2^{(k)}})\big[\mathbf j_2^{(k)},\,\mathbf j_3^{(k)}\big]
\dotsm Q(x_{i_k},x_{j_{{l_k}-1}^{(k)}})\big[\mathbf j_{{l_k}-1}^{(k)},\,\mathbf j_{l_k}^{(k)}\big].\label{ctsts}
\end{align}
Next we define a linear operator 
\begin{align}
&Q(\xi;x_1,x_2,\dots,x_{2n})\notag\\
&\quad:=Q^{(n)}(\xi;x_1,x_2,\dots,x_{2n})Q^{(n-1)}(\xi;x_1,x_2,\dots,x_{2n})\dotsm  Q^{(1)}(\xi;x_1,x_2,\dots,x_{2n}).\label{fxszs}
\end{align}

Recall Remark \ref{rds6qed}. Similarly to \eqref{qdfrqd}, we define a linear functional 
$ \mathbb T^{(2n)}(\xi):V^{\otimes 2n} \to \mathbb C$
 by 
\begin{equation}\label{dmmfvdfc1}
 \mathbb T^{(2n)}(\xi)v_1\otimes v_2\otimes \dotsm \otimes v_{2n}=\langle v_{i_1},v_{\mathbf j_1^{(1)}}\rangle_V  \langle v_{i_2},v_{\mathbf j_1^{(2)}}\rangle_V \dotsm \langle v_{i_n},v_{\mathbf j_1^{(n)}}\rangle_V.
\end{equation}

\begin{theorem}\label{radiaition12345} For $f\in\mathcal H$, let $a^+(f)$ and $a^-(f)$ denote the standard creation and annihilation operators in the $Q$-symmetric Fock space $\mathcal F^Q(\mathcal H)$. Let $\tau(\cdot)=(\cdot\,\Omega,\Omega)_{\mathcal F^Q(\mathcal H)}$ denote the Fock state.  Let $n\in \mathbb N$ and $f_1,f'_1,\dots,f_{2n},f'_{2n}\in \mathcal H$. Then  
\begin{equation}\label{5dhuy8}
\tau\big((a^+(f_1)+a^-(f'_1))\dotsm (a^+(f_{2n-1})+a^-(f'_{2n-1}))\big)=0\end{equation} and 
\begin{align}
&\tau\big((a^+(f_1)+a^-(f'_1))\dotsm (a^+(f_{2n})+a^-(f'_{2n}))\big)
=\sum_{\xi \in \mathcal P^{(2n)}}\int_{X^{2n}}\bigotimes_{\substack{\{i,j\}\in \xi\\i<j}}\sigma^{(2)}(dx_{i}\,dx_{j})\notag\\
&\qquad \times\mathbb T^{(2n)}(\xi)Q(\xi;x_1,\dots,x_{2n})h_1(x_1)\otimes h_2(x_2)\otimes \dotsm \otimes h_{2n}(x_{2n}). \label{fs75w7}
\end{align}
Here $\sigma^{(2)}$ is the measure on $X^2$ satisfying 
\begin{equation}\label{vcyds6ue}
\int_{X^2}f^{(2)}(x_1,x_2)\sigma^{(2)}(dx_1\,dx_2)=\int_X f^{(2)}(x,x)dx 
\end{equation} 
and $h_k(x)=f_k'(x)$ if $k\in I$ and $h_k(x)=f_k(x)$ if $k\in J$.
\end{theorem}

\begin{proof} Denote by $\mathbb F(\mathcal H):=\bigoplus_{n=0}^\infty \mathcal H^{\otimes n}n!$ the full Fock space over $\mathcal H$, and by 
$\mathbb F_{\mathrm{fin}}(\mathcal H)$ the vector space of all sequences $F=(f^{(n)})_{n=0}^\infty\in \mathbb F(\mathcal H)$ such that, for some $N\in\mathbb N$ (depending on $F$), $f^{(n)}=0$ for all $n\ge N$. For 
each $f\in\mathcal H$, we define linear operators $\mathcal A^+(f)$ and $\mathcal A^-(f)$ in $\mathbb F_{\mathrm{fin}}(\mathcal H)$ by 
\begin{align}
&(\mathcal A^+(f)f^{(n)})(x_1,\dots,x_{n+1}):=f(x_1)\otimes f^{(n)}(x_2,\dots,x_{n+1}),\notag\\
&\left(\mathcal A^-(f)f^{(n)}\right)(x_1,\dots,x_{n-1})=\int_X \Big\langle f(x),f^{(n)}(x,x_1,\dots,x_{n-1})\notag\\
&\quad+\sum_{k=1}^{n-1}Q_1(x,x_1)Q_2(x,x_2)\dotsm Q_k(x,x_k)f^{(n)}(x_1,\dots,x_k,x,x_{k+1},\dots,x_{n-1}) \Big\rangle_V\,dx\notag\\
&=\int_X \Big[\operatorname{Tr_1}\big(f(x)\otimes f^{(n)}(x,x_1,\dots,x_{n-1})\big)+\sum_{k=1}^{n-1}\operatorname{Tr_1}\big(Q_2(x,x_1)Q_3(x,x_2)\dotsm Q_{k+1}(x,x_k)\notag\\
&\quad f(x)\otimes f^{(n)}(x_1,\dots,x_k,x,x_{k+1},\dots,x_{n-1})\big)\Big]dx,\quad f^{(n)}\in\mathcal H^{\otimes n}.\label{vcdsawq}
\end{align}

\begin{lemma}\label{jkmnm1235}
For any $f_1,\dots,f_n \in \mathcal H$ and $\sharp_1,\dots,\sharp_n \in \{+,-\}$, we have 
$$
\tau\left(a^{\sharp_1}(f_1)\dotsm a^{\sharp_n}(f_n)\right)=\left(\mathcal A^{\sharp_1}(f_1)\dotsm \mathcal A^{\sharp_n}(f_n)\Omega,\Omega\right)_{\mathbb F(\mathcal H)}.
$$
\end{lemma}

\begin{proof} It is sufficient to prove that, for each $f^{(n)}\in\mathcal H^{\otimes n}$,
\begin{equation}\label{vds}
a^+(f)P_nf^{(n)}=P_{n+1}\mathcal A^+(f)f^{(n)},\quad a^-(f)P_nf^{(n)}=P_{n-1}\mathcal A^-(f)f^{(n)}.
\end{equation}
It easily follows from the definition of the operator $P_n$ that $P_{n+1}(\mathbf 1_{\mathcal H}\otimes P_n)=P_{n+1}$, where $\mathbf 1_{\mathcal H}$ is the identity operator in $\mathcal H$. 
From here the first formula in \eqref{vds} follows.

Next, it follows from the proof of \cite[Theorem~3.1]{BozSpe} that
\begin{equation}\label{vcrtew64}
P_n=\frac1n(\mathbf 1_{\mathcal H}\otimes P_{n-1})(\mathbf 1+U_1+U_1U_2+\dots+U_1U_2\dotsm U_{n-1}).
\end{equation}
Denote by $a^-_{\mathrm{free}}(f)$ the linear operator in $\mathbb F_{\mathrm{fin}}(\mathcal H)$ given by
\begin{equation}\label{vcyrs6u}
(a^-_{\mathrm{free}}(f)f^{(n)})(x_1,\dots,x_{n-1}):=n\int_X\langle f(x),f^{(n)}(x,x_1,\dots,x_{n-1})\rangle_V\,dx,\quad f^{(n)}\in\mathcal H^{\otimes n}.\end{equation}
By \eqref{cxtes5uw}, \eqref{vcdsawq}, \eqref{vcrtew64}, and \eqref{vcyrs6u}, we get, for each $f^{(n)}\in\mathcal H^{\otimes n}$,
\begin{align*}
 a^-(f)P_nf^{(n)}&=a_{\mathrm{free}}^-(f)P_nf^{(n)}\\
 &=a_{\mathrm{free}}^-(f)\frac1n\,(\mathbf 1_{\mathcal H}\otimes P_{n-1})(\mathbf 1+U_1+U_1U_2+\dots+U_1U_2\dotsm U_{n-1})f^{(n)}\\
 &= P_{n-1}\,\frac1n\,a_{\mathrm{free}}^-(f) (\mathbf 1+U_1+U_1U_2+\dots+U_1U_2\dotsm U_{n-1})f^{(n)}\\
 &=P_{n-1}\mathcal A^-(f)f^{(n)}.\qedhere
 \end{align*}
\end{proof}

Formula \eqref{5dhuy8} trivially holds, so we only need to prove \eqref{fs75w7}.
By Lemma \ref{jkmnm1235},
\begin{align}
&\tau\big((a^+(f_1)+a^-(f'_1))\dotsm (a^+(f_{2n})+a^-(f'_{2n}))\big)\notag\\
&\quad=\big(
(\mathcal A^+(f_1)+\mathcal A^-(f'_1))\dotsm(\mathcal A^+(f_{2n})+\mathcal A^-(f'_{2n}))
\Omega,\Omega\big)_{\mathbb F(\mathcal H)}.\label{vsreara}
\end{align}

For each $f \in \mathcal H$ and $k \in \mathbb N$, we define a linear operator $\mathcal A^-(f,k)$ in $\mathbb F_{\mathrm{fin}}(\mathcal H)$  as follows: if $f^{(n)} \in\mathcal H^{\otimes n}$ and $n < k$, then $\mathcal A^-(f,k)f^{(n)}=0$, and if $n \geq k$, then  
\begin{align}
&\left(\mathcal A^-(f,k)f^{(n)}\right)(x_1,\dots,x_{n-1})=\int_X  \operatorname{Tr}_1 Q_2(x,x_1) Q_3(x,x_2)\dotsm Q_k(x,x_{k-1})\notag\\  
&\qquad\quad f(x)\otimes f^{(n)}(x_1,\dots,x_{k-1},x,x_k,\dots,x_{n-1})\,dx. \label{cfyxtszrwa}
\end{align}
Here, for $k=1$, the operator $Q_2(x,x_1)Q_3(x,x_2)\dotsm Q_k(x,x_{k-1})$ is understood as the identity operator. We may equivalently write formula \eqref{cfyxtszrwa} as follows:
\begin{align}
&\left(\mathcal A^-(f,k)f^{(n)}\right)(x_1,\dots, x_{k-1},x_{k+1},\dots,x_{n})=\int_{X^{2}}  \operatorname{Tr}_1 Q_2(x,x_1) Q_3(x,x_2)\dotsm Q_k(x,x_{k-1})\notag\\
&\qquad\quad  f(x)\otimes f^{(n)}(x_1,
\dots,x_{k-1},x_k,x_{k+1},\dots,x_{n})\,\sigma^{(2)}(dx\,dx_k). \label{ctstest}
\end{align}
We will say that the operator $\mathcal A^-(f,k)$ annihilates the $x_k$ variable.

By \eqref{vcdsawq} and \eqref{cfyxtszrwa}, $\mathcal A^-(f)=\sum_{k\geq 1} \mathcal A^-(f,k)$, or equivalently, for each $f^{(n)} \in \mathcal H^{\otimes n}$,
$\mathcal A^-(f)f^{(n)}=\sum_{k=1}^n\mathcal A^-(f,k)f^{(n)}$. Hence, we can express the right-hand side of  \eqref{vsreara} as a summation over all partitions $\xi\in\mathcal P^{(2n)}$ so that each $\xi$ of the form \eqref{ytds}, \eqref{xzseRTYUI} corresponds to the term in which  the creation operators are at places $j_1,\dots,j_n$, the annihilation operators are at places $i_1,\dots,i_n$, and the annihilation operator at place $i_k$ annihilates the variable created by the creation operator at place $j_k$. By  \eqref{xdzazza}--\eqref{nifr75fdrd}, this gives
\begin{align}
&\tau\big((a^+(f_1)+a^-(f'_1))\dotsm (a^+(f_{2n})+a^-(f'_{2n}))\big)\notag\\
&\quad =
\sum_{\xi\in \mathcal P^{(2n)}}
\mathcal A^-(h_{i_n},l_n)\mathcal A^+(h_{i_{n}-1})\dotsm \mathcal A^+(h_{i_{n-1}+1})
\notag\\
&\qquad\mathcal A^-(h_{i_{n-1}},l_{n-1})\mathcal A^+(h_{i_{n-1}+1}) \dotsm\mathcal A^+(h_{i_{n-2}+1})\dotsm \mathcal A^-(h_{i_1},l_1)\mathcal A(h_{i_1+1})\dotsm\mathcal A^+(h_{2n})\Omega. \label{cftz}
\end{align}

For a fixed $\xi \in \mathcal P^{(2n)}$, let us now calculate the value of the expression in the sum  appearing in \eqref{cftz}. To this end, we introduce the following notations. Assume $ \mathbf R$ is a subset of $\{1,\dots,2n\}$ and for some $k \in \{1,\dots,2n\}$ let $\mathbf S=\{k,k+1,\dots,2n\}\setminus \mathbf R$.
We write $\mathbf R=\{\mathbf r_1,\mathbf r_2,\dots,\mathbf r_l\}$ and $\mathbf S=\{\mathbf s_1,\mathbf s_2,\dots,\mathbf s_m\}$ with  $\mathbf s_1<\mathbf s_2<\dotsm<\mathbf s_m$. 
Then, for a function $\psi^{(m)}:X^m \to \mathbb C$, we will use the notation 
\begin{align*}
\psi^{(m)}(x_k,x_{k+1},\dots,x_{2n}\setminus x_{\mathbf r_1},x_{\mathbf r_2},\dots,x_{\mathbf r_l}):=\psi^{(m)}(x_{\mathbf s_1},x_{\mathbf s_2},\dots,x_{\mathbf s_m}).
\end{align*}

Next let $\zeta_1,\zeta_2,\dots,\zeta_{2l}$ be different numbers from the set $\{k,k+1,\dots,2n\}$. Let
$$
\{k,k+1,\dots,2n\}\setminus \{\zeta_1,\zeta_2,\dots,\zeta_{2l}\}=\{\gamma_1,\gamma_2,\dots,\gamma_{2n-k+1-2l}\}
$$
with $\gamma_1<\gamma_2<\dotsm <\gamma_{2n-k+1-2l}$.
We define a linear functional  
$$
\mathbb T^{(2n-k+1)}\left(k;\,\zeta_1,\zeta_2\mid \zeta_3, \zeta_4 \mid \dotsm \mid \zeta_{2l-1},\zeta_{2l}\right):V^{\otimes(2n-k+1)}\to V^{\otimes (2n-k+1-2l)}$$ 
by 
\begin{align*}
&\mathbb T^{(2n-k+1)}\left(k;\,\zeta_1,\zeta_2\mid \zeta_3, \zeta_4 \mid \dotsm \mid \zeta_{2l-1},\zeta_{2l}\right)v_k\otimes v_{k+1}\otimes \dotsm \otimes v_{2n}\\
&\quad=\langle v_{\zeta_1},v_{\zeta_2}\rangle_V \langle v_{\zeta_3},v_{\zeta_4}\rangle_V \dotsm \langle v_{\zeta_{2l-1}},v_{\zeta_{2l}}\rangle_V\, v_{\gamma_1}\otimes v_{\gamma_2}\otimes \dotsm  \otimes v_{2n-k+1-2l}
\end{align*}
for $v_k,\dots,v_n\in V$.

Recall that, for a linear operator $\mathcal C\in\mathcal L(V^{\otimes 2})$ and  $1\leq i <j \leq m$, we defined a linear operator $ \mathcal C[i,j]\in\mathcal L(V^{\otimes m})$. Now, for $k\ge 1$ and $k\le i<j\le k+m-1$, we define an operator $\mathcal C[k;\,i,j]\in\mathcal L(V^{\otimes m})$ by
$$\mathcal C[k;\,i,j]:=\mathcal C[i-k+1,j-k+1].$$
This definition can be interpreted as follows. We enumerate the components of $V^{\otimes m}$ as $k,k+1,\dots,k+m-1$ and the operator $\mathcal C[k;\,i,j]$ acts as the operator $\mathcal C$ on the variables $i$ and $j$. In particular, $\mathcal C[1;i,j]=\mathcal C[i,j]$.

We have 
\begin{align*}
\left(\mathcal A^+(h_{{i_1}+1})\dotsm \mathcal A^+(h_{2n})\Omega \right)(x_{i_1+1},\dots,x_{2n})=h_{i_1+1}(x_{i_1+1}) \otimes \dotsm \otimes h_{2n}(x_{2n})
\end{align*}
and by \eqref{ctstest},
\begin{align*}
&\left(\mathcal A^-(h_{i_1},l_1)\mathcal A^+(h_{i_1+1})\mathcal A^+(h_{2n})\Omega\right)(x_{i_1},\dots,x_{i_{2n}}\setminus x_{i_1},x_{j_1})\\
&\quad=\int_{X^2} \sigma^{(2)}(dx_{i_1}\,dx_{j_1}) \operatorname{Tr}_1 Q_2(x_{i_1},x_{j_1^{(1)}})Q_3(x_{i_1},x_{j_2^{(1)}})\\
&\qquad \dotsm Q_{l_1}(x_{i_1},x_{j_{l_1-1}^{(1)}})h_{i_1}(x_{i_1})\otimes h_{i_2}(x_{i_2})\otimes \dotsm \otimes h_{2n}(x_{2n})\\
&\quad=\int_{X^2} \sigma^{(2)}(dx_{i_1}\,dx_{j_1})\,\mathbb T^{(2n-i_1+1)}(i_1;\,i_1, \mathbf j_1^{(1)})\\
&\qquad Q(x_{i_1},x_{j_1^{(1)}})[i_1;\, \mathbf j_1^{(1)},\mathbf j_2^{(1)}]Q(x_{i_1},x_{j_2^{(1)}})[i_1;\, \mathbf j_2^{(1)},\mathbf j_3^{(1)}]
\dotsm Q(x_{i_1},x_{j_{{l_1}-1}^{(1)}})[i_1;\, \mathbf j_{l_1-1}^{(1)},\mathbf j_{l_1}^{(1)}]\\
&\qquad  h_{i_1}(x_{i_1})\otimes h_{i_2}(x_{i_2})\otimes \dotsm \otimes h_{2n}(x_{2n}).
\end{align*}
Next, we have 
\begin{align*}
&\left(\mathcal A^+(h_{i_2+1})\dotsm \mathcal A^+(h_{i_1-1})\mathcal A^-(h_{i_1},l_1)\mathcal A^+(h_{{i_1}+1})\dotsm \mathcal A^+(h_{2n})\Omega \right)(x_{{i_2}+1},x_{{i_2}+2}\dots,x_{2n}\setminus x_{i_1}, x_{j_1})\notag\\
&\quad=\int_{X^2} \sigma^{(2)}(dx_{i_1}\,dx_{j_1})h_{{i_2}+1}(x_{{i_2}+1})\otimes \dotsm \otimes h_{{i_1}-1}(x_{{i_1}-1})\otimes 
\mathbb T^{(2n-i_1+1)}(i_1;\,i_1, \mathbf j_1^{(1)})\\
&\qquad Q(x_{i_1},x_{j_1^{(1)}})[i_1;\, \mathbf j_1^{(1)},\mathbf j_2^{(1)}]Q(x_{i_1},x_{j_2^{(1)}})[i_1;\, \mathbf j_2^{(1)},\mathbf j_3^{(1)}]
\dotsm Q(x_{i_1},x_{j_{{l_1}-1}^{(1)}})[i_1;\, \mathbf j_{l_1-1}^{(1)},\mathbf j_{l_1}^{(1)}]\\
&\qquad  h_{i_1}(x_{i_1})\otimes g_{i_2}(x_{i_2})\otimes \dotsm \otimes g_{2n}(x_{2n})\\
&\quad=\int_{X^2} \sigma^{(2)}(dx_{i_1}\,dx_{j_1}) \mathbb T^{(2n-i_2)}(i_2+1;\, i_1, \mathbf j_1^{(1)})\\
&\qquad Q(x_{i_1},x_{j_1^{(1)}})[i_2+1;\, \mathbf j_1^{(1)},\mathbf j_2^{(1)}] Q(x_{i_1},x_{j_2^{(1)}})[i_2+1;\,\mathbf j_2^{(1)},\mathbf j_3^{(1)}]\\
&\qquad \dotsm Q(x_{i_1},x_{j_{l_1-1}^{(1)}})[i_2+1;\, \mathbf j_{l_1-1}^{(1)},\mathbf j_{l_1}^{(1)}]
h_{i_2+1}(x_{i_2+1})\otimes h_{i_2+2}(x_{i_2+2})\otimes \dotsm \otimes h_{2n}(x_{2n}).
\end{align*}
From here, using the commutativity of the functionals, respectively operators, acting in different variables, we similarly get:
\begin{align*}
&\big(\mathcal A^-(h_{i_2},l_2)\mathcal A^+(h_{{i_2}+1})\dotsm \mathcal A^+(h_{{i_1}-1})\mathcal A^-(h_{i_1},l_1)\mathcal A^+(h_{{i_1}+1})\\
&\quad 
\dotsm \mathcal A^+(h_{2n})\Omega\big)(x_{i_2},x_{{i_2}+1},\dots,x_{2n} \setminus x_{i_1},x_{j_1},x_{i_2},x_{j_2})\\
&\quad=\int_{X^4}\sigma^{(2)}(dx_{i_2}\,dx_{j_2})\sigma^{(2)}(dx_{i_1}\,dx_{j_1})  \mathbb T^{(2n-i_2+1)}(i_2;\, i_1,\mathbf j_1^{(1)}\mid i_2,\mathbf  j_1^{(2)})\\
&\qquad Q(x_{i_2},x_{j_1^{(2)}})[i_2;\, \mathbf j_1^{(2)},\mathbf j_2^{(2)}]
Q(x_{i_2},x_{j_2^{(2)}})[i_2;\, \mathbf j_2^{(2)},\mathbf j_3^{(2)}]
\dotsm Q(x_{i_2},x_{j_{l_2-1}^{(2)}})[i_2;\,\mathbf j_{l_2-1}^{(2)},\mathbf j_{l_2}^{(2)}] \\
&\qquad Q(x_{i_1},x_{j_1^{(1)}})[i_2;\, \mathbf j_1^{(1)},\mathbf j_2^{(1)}]
Q(x_{i_1},x_{j_2^{(1)}})[i_2;\, \mathbf j_2^{(1)},\mathbf j_3^{(1)}]
\dotsm Q(x_{i_1},x_{j_{l_1-1}^{(1)}})[i_2;\,\mathbf j_{l_1-1}^{(2)},\mathbf j_{l_1}^{(1)}] \\
&\qquad h_{i_2}(x_{i_2}) \otimes h_{{i_2}+1}(x_{{i_2}+1})\otimes \dotsm \otimes h_{2n}(x_{2n}).
\end{align*} 
Continuing by analogy, we get at the $k$th  step:
\begin{align*}
&\big(\mathcal A^-(h_{i_k},l_k)\mathcal A^+(h_{{i_k}+1})\dotsm \mathcal A^+(h_{{i_{k-1}}+1})
\mathcal A^-(h_{i_{k-1}},l_{k-1})\mathcal A^+(h_{i_{k-1}+1})\dotsm A^+(h_{{i_{k-2}}+1})\\
&\qquad \dotsm \mathcal A^-(h_{i_2},l_2)\mathcal A^+(h_{{i_2}+1})\dotsm \mathcal A^+(h_{{i_1}-1})\mathcal A^-(h_{i_1},l_1)\mathcal A^+(h_{{i_1}+1})
\dotsm \mathcal A^+(h_{2n})\Omega\big)\\
&\qquad(x_{i_k},x_{i_k+1},\dots,x_{2n}\setminus x_{i_1},x_{j_1},x_{i_2},x_{j_2},\dots,x_{i_n},x_{j_n})\\
&\quad =\int_{X^{2k}}\sigma^{(2)}(dx_{i_1} dx_{j_1})\sigma^{(2)}(dx_{i_2} dx_{j_2})\dotsm\sigma^{(2)}(dx_{i_k} dx_{j_k})\\
&\qquad \mathbb T^{(2n-i_k+1)}(i_k;i_1,\mathbf j_1^{(1)}\mid i_2,\mathbf j_1^{(2)}\mid\dotsm\mid i_k,\mathbf j_1^{(k)})\\
&\qquad  Q(x_{i_k},x_{j_1^{(k)}})[i_k;\mathbf j_1^{(k)},\mathbf j_2^{(k)}]Q(x_{i_k},x_{j_2 ^{(k)}})[i_k;\mathbf j_2^{(k)},\mathbf j_3^{(k)}]
\dotsm Q(x_{i_k},x_{j_{l_k-1}^{(k)}})
[i_k;\mathbf j_{l_k-1}^{(k)},\mathbf j_{l_k}^{(k)}]\\
&\qquad  Q(x_{i_{k-1}},x_{j_1^{(k-1)}})[i_k;\mathbf j_1^{(k-1)},\mathbf j_2^{(k-1)}]Q(x_{i_{k-1}},x_{j_2^{(k-1)}})[i_k;\mathbf j_2^{(k-1)},\mathbf j_3^{(k-1)}]\\
&\qquad
\dotsm Q(x_{i_{k-1}},x_{j_{l_{k-1}-1}^{(k-1)}})
[i_k;\mathbf j_{l_{k-1}-1}^{(k-1)},\mathbf j_{l_{k-1}}^{(k-1)}]\\
&\qquad \dotsm Q(x_{i_1},x_{j_1^{(1)}})[i_k;\mathbf j_1^{(1)},\mathbf j_2^{(1)}]Q(x_{i_1},x_{j_2 ^{(1)}})[i_k;\mathbf j_2^{(1)},\mathbf j_3^{(1)}]
\dotsm Q(x_{i_1},x_{j_{l_1-1}^{(1)}})
[i_k;\mathbf j_{l_1-1}^{(1)},\mathbf j_{l_1}^{(1)}]\\
&\qquad h_{i_k}(x_{i_k})\otimes h_{i_{k+1}}(x_{{i_k}+1})\otimes \dotsm \otimes h_{2n}(x_{2n}).
\end{align*} 
Thus, after $n$ steps, we get, by using the definitions \eqref{ctsts}--\eqref{dmmfvdfc1},
\begin{align}
&\big(\mathcal A^-(h_{i_n},l_n)\mathcal A^+(h_{i_{n}-1})\dotsm \mathcal A^+(h_{i_{n-1}+1})\mathcal A^-(h_{i_{n-1}},l_{n-1})\mathcal A^+(h_{i_{n-1}+1}) \dotsm\mathcal A^+(h_{i_{n-2}+1})
\notag\\
&\quad\dotsm \mathcal A^-(h_{i_1},l_1)\mathcal A(h_{i_1+1})\dotsm\mathcal A^+(h_{2n})\Omega\big)(x_1,x_2,\dots,x_n)\\
&\quad=\int_{X^{2n}}\sigma^{(2)}(dx_{i_1}\,dx_{j_1})\sigma^{(2)}(dx_{i_2}\,dx_{j_2})\dotsm  \sigma^{(2)}(dx_{i_n}\,dx_{j_n})\notag\\
&\qquad \times \mathbb T^{(2n)}(\xi)Q(\xi;x_1,x_2,\dots,x_{2n})h_1(x_1)\otimes h_2(x_2)\otimes \dotsm \otimes h_{2n}(x_{2n})\label{dhfbvxizkx}.
\end{align} 
Formulas  \eqref{cftz} and \eqref{dhfbvxizkx} imply \eqref{fs75w7}.
\end{proof}

 \begin{remark}
 We can unify formulas \eqref{5dhuy8} and \eqref{fs75w7} into a single formula 
 \begin{align}\label{tdnmcx}
&\tau\big((a^+(f_1)+a^-(f_1'))\dotsm (a^+(f_{n})+a^-(f_{n}'))\big)=\sum_{\xi \in \mathcal P^{(n)}}\int_{X^n}\bigotimes_{\substack{\{i,j\}\in \xi\\i<j}}\sigma^{(2)}(dx_{i}\,dx_{j})\notag\\
&\qquad \times\mathbb T^{(n)}(\xi)Q(\xi;x_1,\dots,x_{n})h_1(x_1)\otimes h(x_2)\otimes \dotsm \otimes h_{n}(x_{n})
 \end{align}
 which holds for all $n\in\mathbb N$.
 Indeed, if $n$ is even, then \eqref{tdnmcx} becomes \eqref{fs75w7} and if $n$ is odd, then the set $\mathcal  P^{(n)}$ is empty, hence the right hand-side of \eqref{tdnmcx} is equal to zero. 
 \end{remark}
 
 Recall that the operator $Q(\xi;x_1,\dots,x_{n})$ was defined through $\xi \in \mathcal P^{(n)}$ and operators $Q(x,x^\prime)$.  But, for $x=(y,z)$ and $x^\prime=(y^\prime,z^\prime)$, we have $Q(x,x^\prime)=Q(y,y^\prime)$. Therefore, we can write the operator $Q(\xi;x_1,\dots,x_n)$ as $Q(\xi;y_1,\dots,y_n)$.

 Similarly to the definition \eqref{vcyds6ue} of the measure $\sigma^{(2)}$ on $X^2$, we define the measure $\nu^{(2)}$ on $Y^2$. 

In accordance with \eqref{buyyr7}, we define $b(f):=a^+(f)+a^-(Jf)$.  We can easily see that Theorem~\ref{radiaition12345}  implies the following 
 
 \begin{corollary}\label{wsdfrcgvddf}
 Let $\tau$ be the Fock state.
 For any $g_1, g_2 \in \mathcal G$, we define a complex-valued measure $\lambda^{(2)} [g_1,g_2]$ on $Y^2$ by
 \begin{equation}\label{b12nmvbn}
 \lambda^{(2)} [g_1,g_2](dy_1\,dy_2)=(g_2,g_1)_{\mathcal G}\, \nu^{(2)}(dy_1\, dy_2). 
 \end{equation}
  Then, for any $\varphi_1,\dots,\varphi_n \in C_0(Y; V_{\mathbb R})$ and $g_1,\dots,g_n\in \mathcal G$, we have 
 \begin{align*}
 \tau\big(b(\varphi_1\otimes g_1)\dotsm b(\varphi_n \otimes g_n)\big)&=\sum_{\xi \in \mathcal P^{(n)}}\int_{Y^n}\bigotimes_{\substack{\{i,j\}\in \xi\\i<j}} \lambda^{(2)}[g_i,g_j](dy_i\, dy_j)\\ 
 &\quad\mathbb T^{(n)}(\xi)Q(\xi;y_1,\dots,y_n)\varphi_1(y_1)\otimes \varphi_2(y_2)\otimes \dotsm \otimes \varphi_n(y_n).
 \end{align*}
 \end{corollary}
 
  Below we will also need a formula for 
 \begin{align}\label{tangordsxvc652}
 \tau\left(a^-(f_1)\cdots a^-(f_n)a^+(f_{n+1})\cdots a^+(f_{2n})\right)
 \end{align}
    This is, of course, a special case of formula \eqref{fs75w7}, which can be simplified in the case of~\eqref{tangordsxvc652}.
  
  For each permutation $\pi \in S_n$, we define a partition $\xi\in \mathcal P^{(2n)}$ as follows: 
  $$\xi=\big\{\{1,n+\pi(n)\},\{2,n+\pi(n-1)\},\dots, \{n,n+\pi(1)\}\big\}.$$
   We denote by $\mathcal S_n$ the subset of $\mathcal P^{(2n)}$ consisting of all such partitions. Equivalently, $\mathcal S_n$ consists of all partitions $\xi\in \mathcal P^{(2n)}$ such that, for  each $\{i,j\}\in\xi$ with $i<j$, we have $i\le n$ and $j\geq n+1$. 
  The following corollary is immediate.
  
 \begin{corollary}\label{xkiasaxzs}
 Let $\tau$ be the Fock state. 
Let $\varphi_1,\dots,\varphi_{2n}\in C_0(Y; V)$ and $g_1,\dots, g_{2n}\in \mathcal G$.  Let the complex-valued measure $\lambda^{(2)}[g_i,g_j]$ be defined by  \eqref{b12nmvbn}.
We have 
\begin{align*}
&\tau\big(a^-(\varphi_1\otimes g_1)\cdots a^-(\varphi_n\otimes g_n)a^+(\varphi_{n+1}\otimes g_{n+1})\cdots a^+(\varphi_{2n}\otimes g_{2n})\big)\\
&\quad=\sum_{\xi \in \mathcal S_n} \int_{Y^{2n}} \bigotimes_{\substack{\{i,j\}\in \xi\\i<j}}\lambda^{(2)}[g_i,g_j](dy_i\,dy_j)\\
&\qquad\times \mathbb T^{(2n)}Q(\xi;y_1,\dots,y_{2n})\varphi_1(y_1)\otimes \varphi_2(y_2)\otimes \cdots\otimes \varphi_{2n}(y_{2n}). 
\end{align*}
  Here $\mathbb T^{(2n)}:V^{\otimes 2n} \to \mathbb C$ is the linear functional defined by \eqref{qdfrqd} and the linear operator $Q(\xi;x_1,\dots,x_{2n})$ is defined by formula \eqref{fxszs} in which 
\begin{align*}
Q^{(k)}(\xi;x_1,\dots,x_{2n})=&Q_{n+k}(x_{n-k+1},x_{j_1^{(k)}})Q_{n+k+1}(x_{n-k+1},x_{j_2^{(k)}})\\
&\quad\cdots Q_{n+k+l_k-1}(x_{n-k+1},x_{j_{l_k-1}^{(k)}}),
\end{align*}
where, for $\xi=\big\{\{1,j_n\},\{2,j_{n-1}\},\dots,\{n,j_1\}\big\}\in\mathcal S_n$, we denote $J=\{j_1,j_2,\dots,j_n\}$ and 
\begin{align*}
J^{(k)}&=\big\{j\in J \mid j\le j_k,\, j\neq j_1,\dots, j\neq j_{k-1}\big\},\\
&=\big\{j_1^{(k)},j_2^{(k)},\dots, j_{l_k}^{(k)}\big\},\quad j_1^{(k)}<j_2^{(k)}<\dots<j_{l_k}^{(k)}.
\end{align*}
 \end{corollary} 

 \section{Quasi-free states on the $Q$-MCR algebra}\label{kljiun6958}

Let $\tau$ be a state on the $Q$-MCR algebra  $\mathbb A$. We assume that, for any $\sharp_1,\dots,\sharp_n\in \{+,-\}$ and any  $g_1,\dots,g_n\in\mathcal G$, the functional  $\tau^{(n)}_{\sharp_1,\dots,\sharp_n }(\cdot,g_1,\dots,g_n)$ is continuous on $C_0(Y^n; V^{\otimes n})$. In view of Corollaries \ref{wsdfrcgvddf} and \ref{xkiasaxzs}, we now give the following

\begin{definition} 
(i) Assume that, for any $g_1,g_2\in \mathcal G$, there exists a complex-valued measure $\lambda^{(2)}[g_1,g_2]$ on $Y^2$ such that, for all $\varphi_1,\varphi_2 \in C_0(Y; V_{\mathbb R})$,
\begin{align*}
\tau\big(B(\varphi_1\otimes g_1)B(\varphi_2\otimes g_2)\big)=\int_{Y^2}\langle \varphi_1(y_1),\varphi_2(y_2)\rangle_V\, \lambda^{(2)}[g_1,g_2](dy_1\,dy_2).
\end{align*}
  We say that $\tau$ is a {\it strongly quasi-free state}, if all $\varphi_1,\dots,\varphi_n\in C_0(Y; V_{\mathbb R})$ and $g_1,\dots,g_n \in \mathcal G$, we have 
 \begin{align}
 &\tau\big(B(\varphi_1\otimes g_1)\dotsm B(\varphi_n \otimes g_n)\big)=\sum_{\xi \in \mathcal P^{(n)}}\int_{Y^n}\bigotimes_{\substack{\{i,j\}\in \xi\\i<j}} \lambda^{(2)}[g_i,g_j](dy_i\, dy_j)\notag\\
 &\qquad \times \mathbb T^{(n)}(\xi)Q(\xi;y_1,\dots,y_n) \varphi_1(y_1)\otimes \varphi_2(y_2)\otimes \dotsm \otimes \varphi_n(y_n).\label{dseq4reure}
 \end{align}

(ii) Assume that, for any $g_1,g_2\in \mathcal G$, there exists a complex-valued measure $\rho^{(2)}[g_1,g_2]$ on $Y^2$ such that, for all $\varphi_1,\varphi_2 \in C_0(Y; V)$,
\begin{align*}
\tau\left(A^+(\varphi_1\otimes g_1)A^-(\varphi_2\otimes g_2)\right)=\int_{Y^2}\rho^{(2)}[g_1,g_2](dy_1\,dy_2)\langle \varphi_1(y_1),\varphi_2(y_2)\rangle_V.
\end{align*}
We say that $\tau$ is a {\it gauge-invariant quasi-free state} if for all $m,n\in \mathbb N$, $\varphi_1,\dots,\varphi_{m+n}\in C_0(Y;V)$ and $g_1,\dots, g_{m+n} \in \mathcal G$, we have
\begin{align}
&\tau\left(A^+(\varphi_1\otimes g_1)\dotsm A^+(\varphi_m\otimes g_m)A^-(\varphi_{m+1}\otimes g_{m+1})\dotsm A^-(\varphi_{m+n}\otimes g_{m+n})\right)\notag\\
&\quad=\delta_{m,n}\sum_{\xi \in \mathcal S_n} \int_{Y^{2n}} \bigotimes_{\substack{\{i,j\}\in \xi\\i<j}}\rho^{(2)}[g_i,g_j](dy_i\,dy_j)\notag\\
&\qquad \times \mathbb T^{(2n)}Q(\xi;y_1,\dots,y_{2n})\varphi_1(y_1)\otimes \varphi_2(y_2)\otimes \cdots\otimes \varphi_{2n}(y_{2n}). \label{bugdrt6ew5a5u}
\end{align}
In formula \eqref{bugdrt6ew5a5u}, $\delta_{m,n}$ denotes the Kronecker delta.

\end{definition}

\begin{remark}
Both formulas \eqref{dseq4reure} and \eqref{bugdrt6ew5a5u} imply that the state $\tau$ is completely determined by the operator-valued function $Q(y_1,y_2)$ and by the values   
$$\tau\big(B(\varphi_1\otimes g_1)B(\varphi_2\otimes g_2)\big)\text{ and }\tau\left(A^+(\varphi_1\otimes g_1)A^-(\varphi_2\otimes g_2)\right),$$ respectively. 
\end{remark}

Corollary \ref{wsdfrcgvddf} implies the following 

\begin{proposition}
The Fock state on the $Q$-MCR algebra is strongly quasi-free and the corresponding measure $\lambda^{(2)}[g_1,g_2](dy_1\,dy_2)$ is given by formula \eqref{b12nmvbn}. The Fock state is gauge-invariant quasi-free, with  $\rho^{(2)}[g_1,g_2](dy_1\,dy_2)=0$.
\end{proposition}

We will now construct a class of gauge-invariant quasi-free states and strongly quasi-free states  on the $Q$-MCR algebra. We will be able to do this under additional assumptions on the operator-valued function $Q$.

\begin{assumptions}\label{Rosated} We assume:

\begin{itemize}
\item[(i)] For all $y_1,y_2\in Y$, $\widetilde Q(y_1,y_2)$ is unitary and 
$
\widetilde Q(y_1,y_2)^*=\widetilde Q(y_2,y_1)$.

\item[(ii)] For all $y_1,y_2\in Y$, we have $\widetilde {\widetilde {Q}}(y_1,y_2)=Q(y_1,y_2)$.

\item[(iii)] For all $y_1,y_2\in Y$,  we have $\widehat Q(y_1,y_2)=Q(y_1,y_2)$ and $\widehat{\widetilde Q}(y_1,y_2)=\widetilde Q(y_1,y_2)$.

\item[(iv)] There exists a constant $\varkappa\in \mathbb R$ such that, for all $y \in Y$ and $v^{(2)} \in V^{\otimes 2}$, we have $\operatorname{Tr}\big(\widetilde Q(y,y)v^{(2)}\big)=\varkappa \operatorname{Tr}v^{(2)}$.

\item[(v)] Let $Y_1$ and $Y_2$ be two copies of $Y$ and let $\mathbf Y:=Y_1\sqcup Y_2$. Define a (continuous) function $\mathbf Q:\mathbf  Y^2\to \mathcal L\left(V^{\otimes 2}\right)$ as follows: ${\bf Q}(y_1,y_2):=Q(y_1,y_2)$ if either  $y_1,y_2\in Y_1$ or $y_1,y_2 \in Y_2$, and $\mathbf Q(y_1,y_2):=\widetilde Q(y_2,y_1)$ if either $y_1\in Y_1$, $y_2\in Y_2$ or  $y_1\in Y_2$, $y_2 \in Y_1$. [Note that, by (i) and (ii), for any $(y_1,y_2)\in \mathbf Y^2$, $\mathbf Q(y_1,y_2)$ is a unitary operator and $\mathbf Q(y_1,y_2)^*=\mathbf Q(y_2,y_1)$.] Then the function $\mathbf Q:\mathbf Y^2\to \mathcal L\left(V^{\otimes 2}\right)$  satisfies the functional Yang--Baxter equation point-wise, cf.~\eqref{v6w6u34}. 
\end{itemize}
\end{assumptions}

Similarly to Assumption~\ref{Rosated}, (v), we define $\mathbf X:=X_1\sqcup X_2$ and let $\mathbf Q:\mathbf X^2\to\mathcal L(V^{\otimes 2})$ be defined by $\mathbf Q(x_1,x_2):=\mathbf Q(y_1,y_2)$, where $x_i=(y_i,z_i)\in \mathbf X$, $i=1,2$. For 
$$\mathbf f\in L^2(\mathbf X;V)=L^2(X_1;V)\oplus L^2(X_2;V)= \mathcal H\oplus\mathcal H,$$
 we construct the standard creation operator $\mathbf a^+(\mathbf f)$ and the standard annihilation operator $\mathbf a^-(\mathbf f)$ in the $\mathbf Q$-symmetric Fock space $\mathcal F^{\mathbf Q}(L^2(\mathbf X;V))$. 
Similarly to Section~\ref{fgdyjdde}, we then construct the $\mathbf Q$-MCR algebra, denoted by $\mathbf A$, and the Fock state on $\mathbf A$, denoted by $ \boldsymbol\tau$. In particular, we think of $\mathbf A$ as the algebra spanned by operator-valued integrals 
\begin{equation}\label{pppooojjj}
\int_{{\bf X}^n}\langle \mathbf f^{(n)}(x_1,\dots,x_n),{\bf a}^{\sharp_1}(x_1)\otimes \dotsm \otimes {\bf a}^{\sharp_n}(x_n)\rangle_{V^{\otimes n}}\, dx_1\dotsm dx_n,
\end{equation}
where $\sharp_1,\dots,\sharp_n \in \{+,-\}$ and $\mathbf f^{(n)}\in \boldsymbol{\mathfrak F}^{(n)}$, where the space $\boldsymbol{\mathfrak F}^{(n)}$ is constructed similarly to $\mathfrak F^{(n)}$ but  by starting with $\mathbf X$ instead of $X$, compare with  \eqref{vtrs56e57}

Let $\sharp_1,\dots,\sharp_n \in \{+,-\}$, $i_1,\dots,i_n\in\{1,2\}$ and let $f^{(n)}\in \mathfrak F^{(n)}$. We define
\begin{equation}\label{fdea4q567i8o}
\int_{ X^n}\langle  f^{(n)}(x_1,\dots,x_n), \mathbf a_{i_1}^{\sharp_1}(x_1)\otimes \dotsm \otimes \mathbf a_{i_n}^{\sharp_n}(x_n)\rangle_{V^{\otimes n}}\, dx_1\dotsm dx_n
\end{equation}
to be equal to the operator-valued integral \eqref{pppooojjj} in which 
\begin{equation}
\mathbf f^{(n)}(x_1,\dots,x_n)= 
\begin{cases}
f^{(n)}(x_1,\dots,x_n),& \text{if }(x_1,\dots,x_n)\in X_{i_1} \times X_{i_2}\times \dotsm \times X_{i_n},\\
 0,&  \text{otherwise}.
\end{cases}\label{fdy6eio}
\end{equation}
In formula \eqref{fdy6eio}, we identified $(x_1,\dots,x_n)\in X_{i_1} \times X_{i_2}\times \dotsm \times X_{i_n}$ with  $(x_1,\dots,x_n)\in X^n$. Thus, intuitively, for $i=1,2$ and $\sharp\in \{+,-\}$, $\mathbf a^\sharp_i(\cdot)$ is the restriction of $\mathbf a^\sharp(\cdot)$ to $X_i$.

As easily seen, the algebra $\mathbf A$ is spanned by the operator-valued integrals of the form \eqref{fdea4q567i8o}.
Similarly to \eqref{cftrs5y64e}, we denote the integral in \eqref{fdea4q567i8o} by 
$\boldsymbol\Phi(f^{(n)};i_1,\sharp_1,\dots,i_n,\sharp_n)$.

Formulas \eqref{gdtrs5u4}--\eqref{cdtesw5w}, Proposition~\ref{zzzzaaqqq} and Assumption~\ref{Rosated} imply

\begin{lemma}\label{vye67gytr56}
 Let $f^{(n)}\in\mathfrak F^{(n)}$, $i_1,\dots,i_n\in\{1,2\}$, $\sharp_1,\dots,\sharp_n\in\{+,-\}$, and $k\in\{1,\dots,n-1\}$. The following commutation relations hold:
 if $\sharp_k=\sharp_{k+1}$ and $i_k=i_{k+1}$,
$$\boldsymbol\Phi(f^{(n)};i_1,\sharp_1,\dots,i_n,\sharp_n)=\boldsymbol\Phi(Q_k(x_k,x_{k+1})f^{(n)}(x_1,\dots,x_{k+1},x_k,\dots,x_n);i_1,\sharp_1,\dots,i_n,\sharp_n);$$
if $\sharp_k=\sharp_{k+1}$ and $i_k\ne i_{k+1}$,
\begin{align*}
&\boldsymbol\Phi(f^{(n)};i_1,\sharp_1,\dots,i_n,\sharp_n)\\
&\quad =\boldsymbol\Phi(\widetilde Q_k(x_{k+1},x_k)f^{(n)}(x_1,\dots,x_{k+1},x_k,\dots,x_n);i_1,\sharp_1,\dots,i_{k+1},\sharp_{k},i_{k},\sharp_{k+1},\dots, i_n,\sharp_n);
\end{align*}
 if $\sharp_k=-$, $\sharp_{k+1}=+$ and $i_k=i_{k+1}$,
\begin{align*}
&\boldsymbol\Phi(f^{(n)};i_1,\sharp_1,\dots,i_n,\sharp_n)\\
&\quad =\boldsymbol\Phi(\widetilde Q_k(x_{k+1},x_k)f^{(n)}(x_1,\dots,x_{k+1},x_k,\dots,x_n);i_1,\sharp_1,\dots,i_{k},\sharp_{k+1},i_{k+1},\sharp_{k},\dots, i_n,\sharp_n)\\
&\qquad+\boldsymbol\Phi(g^{(n-2)};i_1,\sharp_1,\dots,i_{k-1},\sharp_{k-1},i_{k+2},\sharp_{k+2},\dots,i_n,\sharp_n),
\end{align*}
where $g^{(n-2)}$ is given by \eqref{cdtesw5w};  if $\sharp_k=-$, $\sharp_{k+1}=+$ and $i_k\ne i_{k+1}$,
\begin{align*}
&\boldsymbol\Phi(f^{(n)};i_1,\sharp_1,\dots,i_n,\sharp_n)\\
&\quad =\boldsymbol\Phi(Q_k(x_{k},x_{k+1})f^{(n)}(x_1,\dots,x_{k+1},x_k,\dots,x_n);i_1,\sharp_1,\dots,i_{k+1},\sharp_{k+1},i_{k},\sharp_{k},\dots, i_n,\sharp_n);
\end{align*}
if $\sharp_k=+$, $\sharp_{k+1}=-$ and $i_k=i_{k+1}$,
\begin{align*}
&\boldsymbol\Phi(f^{(n)};i_1,\sharp_1,\dots,i_n,\sharp_n)\\
&\quad =\boldsymbol\Phi(\widetilde Q_k(x_{k+1},x_k)f^{(n)}(x_1,\dots,x_{k+1},x_k,\dots,x_n);i_1,\sharp_1,\dots,i_{k},\sharp_{k+1},i_{k+1},\sharp_{k},\dots, i_n,\sharp_n)\\
&\qquad-\varkappa\boldsymbol\Phi(g^{(n-2)};i_1,\sharp_1,\dots,i_{k-1},\sharp_{k-1},i_{k+2},\sharp_{k+2},\dots,i_n,\sharp_n),
\end{align*}
where $g^{(n-2)}$ is given by \eqref{cdtesw5w}; if $\sharp_k=+$, $\sharp_{k+1}=-$ and $i_k\ne i_{k+1}$,
\begin{align*}
&\boldsymbol\Phi(f^{(n)};i_1,\sharp_1,\dots,i_n,\sharp_n)\\
&\quad =\boldsymbol\Phi(Q_k(x_{k},x_{k+1})f^{(n)}(x_1,\dots,x_{k+1},x_k,\dots,x_n);i_1,\sharp_1,\dots,i_{k+1},\sharp_{k+1},i_{k},\sharp_{k},\dots, i_n,\sharp_n).
\end{align*}
\end{lemma}

Let us fix a bounded linear operator $K\in \mathcal L\left(\mathcal G\right)$. In the case where $\varkappa\geq 0$, we assume that $K\geq 0$ and in the case where $\varkappa<0$, we assume $ 0\le K\le- \frac{1}{\varkappa} $. Let 
$K_1:=\sqrt K$ and $K_2:=\sqrt{1+\varkappa K}$. Furthermore, we  assume that both operators $K_1$ and $K_2$ are not equal to zero, equivalently $K\ne0$ and $K\ne-\frac{1}{\varkappa}$.

For a bounded linear operator  $C \in \mathcal L(\mathcal G)$, we denote $C':=JCJ$, the complex conjugate of $C$.
Here, just as above $J$ is the antilinear operator of complex conjugation.  

Let $f=\varphi\otimes g$ with $\varphi\in C_0(Y; V)$ and $g \in \mathcal G$. We define 
\begin{align}
&A^+(f)=\int_X\langle \varphi(y)(K_2g)(z),{\bf a}_2^+(x)\rangle_V\, dx+\int_X \langle \varphi(y)(K_1g)(z),{\bf a}_1^-(x)\rangle_V\, dx,\notag\\
&A^-(f)=\int_X\langle \varphi(y)(K_2^\prime g)(z),{\bf a}_2^-(x)\rangle_V\, dx+\int_X \langle \varphi(y)(K_1^\prime g)(z),{\bf a}_1^+(x)\rangle_V\, dx.\label{siisssaass}
\end{align}
We define the corresponding operator-valued distributions $A^+(x)$, $A^-(x)$ through the formula 
\begin{equation}\label{vcfstsj7k6r78i}
 A^+(f)=\int_X \langle f(x), A^+(x)\rangle_V\, dx,\quad A^-(f)=\int_X \langle f(x),A^-(x) \rangle_V\, dx.
 \end{equation}

\begin{proposition}\label{whisk63254}
We have $\left(A^+(f)\right)^*=A^-(Jf)$ and the operators $A^+(f)$, $A^-(f)$ defined by \eqref{siisssaass} satisfy the $Q$-MCR, see \eqref{xseas5aq53q}--\eqref{yqdqidi}.
\end{proposition}

\begin{proof}
The proposition is a consequence of \eqref{siisssaass}, Lemma~\ref{vye67gytr56}, and the following observation: for any $g_1,g_2\in\mathcal G$,
\begin{align*}
&\int_Z (K_2'g_1)(z)(K_2g_2)(z)dz-\varkappa\int_Z(K_1'g_1)(z)(K_1g_2)(z)dz\\
&\quad=(K_2g_2,K_2Jg_1)_\mathcal G-\varkappa(K_1g_2,K_1Jg_1)_\mathcal G\\
&\quad=\big((K_2^2-\varkappa K_1^2)g_2,Jg_1\big)_\mathcal G=\int_Z g_1(z)g_2(z)dz.\qedhere
\end{align*}

\end{proof}

Inspired by Proposition~\ref{whisk63254}, we will now construct a $Q$-MCR algebra $\mathbb A$. Formally, $\mathbb A$ is spanned by the operators
$$\Phi(f^{(n)};\sharp_1,\dots,\sharp_n):=\int_{X^n}\langle f^{(n)}(x_1,\dots,x_n), A^{\sharp_1}(x_1)\otimes \dotsm \otimes  A^{\sharp_n}(x_n) \rangle_{V^{\otimes n}} \,dx_1\dotsm dx_n,$$
where $f^{(n)}\in\mathfrak F^{(n)}$,  $\sharp_1,\dots,\sharp_n\in\{+,-\}$. We now rigorously define $\Phi(f^{(n)};\sharp_1,\dots,\sharp_n)$  as follows.

Denote $K(+,1):=K_1$, $K(+,2):=K_2$, $K(-,1):=K_1'$, $K(-,2):=K_2'$, and denote $\gamma(+,1):=-$, $\gamma(+,2):=+$ $\gamma(-,1):=+$, $\gamma(-,2):=-$. Let $f^{(n)}\in\mathfrak F^{(n)}$ be of the form \eqref{tsar43}. Then we define
\begin{align}
&\Phi(f^{(n)};\sharp_1,\dots,\sharp_n):=\sum_{i_1,\dots,i_n\in\{1,2\}}\boldsymbol\Phi\big(\varphi^{(n)}(y_1,\dots,y_n)\notag\\
&\quad\times\left(K(\sharp_1,i_1)g_1\right)(z_1)\dotsm  \left(K(\sharp_n,i_n)g_n\right)(z_n);i_1,\gamma(\sharp_1,i_1),\dots, i_n,\gamma(\sharp_n,i_n)\big).\label{vydstesa5qw5w6}
\end{align}
Next, we extend the definition of $\Phi(f^{(n)};\sharp_1,\dots,\sharp_n)$  by linearity to the case of an arbitrary $f^{(n)}\in\mathfrak F^{(n)}$. 

By construction, each $\Phi(f^{(n)};\sharp_1,\dots,\sharp_n)$ belongs to the $\mathbf Q$-MCR algebra $\mathbf A$. Let $\mathbb A$ denote the unital $*$-algebra generated by these elements. Thus, $\mathbb A$ is a $*$-sub-algebra of ${\bf A}$. We will denote by $\tau$ the restriction of the Fock state $\boldsymbol \tau$ to $\mathbb A$.

\begin{theorem}\label{buty7e6u} (i)
The $*$-algebra $\mathbb A$ constructed above is a $Q$-MCR algebra, i.e., it satisfies the conditions \eqref{tera4qy}--\eqref{cdtesw5w}. 

(ii) The state $\tau$ on $\mathbb A$ is gauge-invariant quasi-free and the corresponding complex-valued measure $\rho^{(2)}[g_1,g_2]$ on $Y^2$ is given by 
\begin{equation}\label{convoadsxc}
\rho^{(2)}[g_1,g_2](dy_1 dy_2)=\nu^{(2)}(dy_1 dy_2)\int_Z(Kg_1)(z)g_2(z)\mu(dz).
\end{equation}

(iii) The state $\tau$ is strongly quasi-free if and only if 
\begin{equation}\label{vutde7i}
\widetilde Q(y_1,y_2)=Q(y_2,y_1).
\end{equation}
 In the latter case, the corresponding complex-valued measure $\lambda^{(2)}[g_1,g_2]$ on $Y^2$ is given by 
\begin{equation}\label{vgdyr8t8tr8rd}
\lambda^{(2)}\left[g_1,g_2\right](dy_1\,dy_2)= \nu^{(2)}(dy_1\,dy_2)\big((g_2,g_1)_{\mathcal G}+(Kg_1,g_2)_{\mathcal G}+\varkappa (g_2,Kg_1)_{\mathcal G}\big).
\end{equation} 
\end{theorem}

\begin{proof} (i) By using \eqref{vydstesa5qw5w6}, one shows that $\mathbb A$ is a $Q$-MCR algebra similarly to the proof of  Proposition~\ref{whisk63254}.

(ii) For any $(f_1,f_2)\in\mathcal H\oplus\mathcal H$, we denoted by 
$\mathbf a^\sharp(f_1,f_2)=\mathbf a_1(f_1)+\mathbf a_2(f_2)$ ($\sharp\in\{+,-\}$) the corresponding standard creation and annihilation operators in the $\mathbf Q$-symmetric Fock space $\mathcal F^\mathbf Q(\mathcal H\oplus\mathcal H)$.  Let $f=\varphi\otimes g$ with $\varphi\in C_0(Y; V)$ and $g \in \mathcal G$. By \eqref{siisssaass},
\begin{align*}
A^+(f)&=\mathbf a_2^+(\varphi\otimes (K_2g))+\mathbf a_1^-(\varphi\otimes (K_1g)),\\
A^-(f)&=\mathbf a_2^-(\varphi\otimes (K_2'g))+\mathbf a_1^+(\varphi\otimes (K_1'g)).
\end{align*}

For $f_1=\varphi_1\otimes g_1$ and $f_2=\varphi_2\otimes g_2$, we have
\begin{align}
\tau\big(A^+(f_1)A^-(f_2)\big)&=\tau\big(\mathbf a_1^-(\varphi_1\otimes (K_1g_1))\,\mathbf a_1^+(\varphi_2\otimes (K_1'g_2))\big)\notag\\
& =\int_Y\langle \varphi_1(y),\varphi_2(y)\rangle_V\,dy\,(K_1g_1,K_1Jg_2)_\mathcal G\notag\\
&=\int_{Y^2}\langle \varphi_1(y_1),\varphi_2(y_2)\rangle_V\, \rho^{(2)}[g_1,g_2](dy_1\,dy_2),\label{cgxts5wu5}
\end{align}
where $\rho^{(2)}[g_1,g_2]$ is given by \eqref{convoadsxc}. Next, for any $f_i=\varphi_i\otimes g_i$ ($i=1,\dots,m+n$), we have
\begin{align}
&\tau\big(A^+(f_1)\dotsm A^+(f_m)A^-(f_{m+1})\dotsm A^-(f_{m+n})\big)=\tau\big(\mathbf a_1^-(\varphi_1\otimes (K_1g_1))\notag\\
 &\qquad\dotsm \mathbf a_1^-(\varphi_m\otimes (K_1g_m))\,
\mathbf a_1^+(\varphi_{m+1}\otimes (K_1'g_{m+1}))\dotsm \mathbf a_1^+(\varphi_{m+n}\otimes (K_1'g_{m+n})\big)\notag\\
&\quad =\tau\big(a^-(\varphi_1\otimes (K_1g_1))\dotsm a^-(\varphi_m\otimes (K_1g_m))\,
a^+(\varphi_{m+1}\otimes (K_1'g_{m+1}))\notag\\
&\qquad\dotsm a^+(\varphi_{m+n}\otimes (K_1'g_{m+n})\big).\label{cfstaa}
\end{align}
In the last equality of \eqref{cfstaa},  $a^+(\cdot)$ and $a^-(\cdot)$ denote the standard creation and annihilation operators in the $Q$-symmetric Fock space $\mathcal F^Q(\mathcal H)$, and $\tau$ denotes the corresponding Fock state. Now the statement (ii) of the theorem  follows from \eqref{cgxts5wu5}, \eqref{cfstaa}, and Corollary~\ref{xkiasaxzs}.

(iii) For $f=\varphi\otimes g$ with $\varphi\in C_0(Y; V_\R)$ and $g \in \mathcal G$, the operator $B(f)=A^+(f)+A^-(Jf)$ is given by 
$$B(f)=\mathbf a^+(\varphi\otimes(JK_1g),\varphi\otimes(K_2g))+\mathbf a^-(\varphi\otimes(K_1g),\varphi\otimes (JK_2g)).$$
Therefore, for such $f_1=\varphi_1\otimes g_1$ and $f_2=\varphi_1\otimes g_2$,
\begin{align}
&\tau\big(B(f_1)B(f_2)\big)\notag\\
&\quad=\int_Y\langle \varphi_1(y),\varphi_2(y)\rangle_V\,dy\,\bigg(\int_Z (K_1g_1)(z)(JK_1g_2)(z)dz+\int_Z (JK_2g_1)(z)(K_2g_2)(z)dz\bigg)\notag\\
&\quad=\int_Y\langle \varphi_1(y),\varphi_2(y)\rangle_V,dy\,\big((K_1g_1,K_1g_2)_\mathcal G+(K_2g_2,K_2g_1)_\mathcal H\big)\notag\\
&\quad
=\int_{Y^2}\langle \varphi_1(y_1),\varphi_2(y_2)\rangle_V\, \lambda^{(2)}[g_1,g_2](dy_1\,dy_2),\label{vrts6uews}
\end{align}
where $ \lambda^{(2)}[g_1,g_2]$ is given by \eqref{vgdyr8t8tr8rd}.

Assume that condition \eqref{vutde7i} is not satisfied.  Then, by a straightforward calculation of $\tau\big(B(f_1)B(f_2)B(f_3)B(f_4)\big)$, we see that formula  \eqref{dseq4reure} fails for the state $\tau$ when $n=4$. Hence, $\tau$ is not strongly gauge-invariant. 

On the other hand, if condition \eqref{vutde7i} is satisfied, then for all $(y_1,y_2)\in\mathbf Y^2$ we have $\mathbf Q(y_1,y_2)=Q(y_1,y_2)$, compare with Assumption~\ref{Rosated} (v). Then that the statement of the theorem about $\tau$  being strongly quasi-free  follows from \eqref{vrts6uews} and a straightforward analog of Corollary~\ref{wsdfrcgvddf} when $X$ is replaced by $\mathbf X$.
\end{proof}

\section{Examples}\label{vcrte6u}

We will now consider several  classes of examples of the operator-valued function $Q$ to which our construction of quasi-free states is applicable. 

\subsection{Abelian anyons}

Let $V=\mathbb C$, which implies that $Q:Y^2\to\mathbb C$ is a continuous complex-valued function satisfying $|Q(y_1,y_2)|=1$ and $\overline{Q(y_1,y_2)}=Q(y_2,y_1)$. Assumption~\ref{Rosated} is now trivially satisfied, with $\varkappa =Q(y,y)\in\{-1,1\}$ and   $\widetilde Q(y_1,y_2)=Q(y_1,y_2)$. Thus, for each appropriate choice of a bounded linear operator $K$ in $\mathcal G$, we obtain the corresponding gauge-invariant quasi-free state $\tau$ on the $*$-algebra $\mathbb A$ of the $Q$-anyon commutation relations. This is essentially the construction from  \cite{anyons}. Note that, in the latter paper, the function $Q$ was not assumed to be continuous, which led to an arbitrary choice of the value of $\varkappa =Q(y,y)\in\R$, and to a different definition of the counterpart of the space $\mathfrak F^{(n)}$. 

Unless $Q(y_1,y_2)$ is not identically equal to 1 (which is the case of the CCR) or $-1$ (which is the case of the CAR), the function $Q$ does not satisfy condition \eqref{vutde7i}, hence the state $\tau$  is not strongly quasi-free.

\subsection{Lifting of  anyon commutation relations to multicomponent commutation relations}

Let $V=\mathbb C^r$ and, for each $i,j\in\{1,\dots,r\}$, we fix  a continuous complex-valued function $q(y_1,y_2,i,j)$ on $Y^2$ such that $|q(y_1,y_2,i,j)|=1$ and 
$$\overline{q(y_1,y_2,i,j)}=q(y_2,y_1,j,i).$$ 
We also assume that there exists a $\varkappa\in\{-1,1\}$ such that 
$q(y,y,i,i)=\varkappa$ for all $y\in Y$ and $i\in\{1,\dots,r\}$. 

For $(y_1,y_2)\in Y^2$, we define $Q(y_1,y_2)\in \mathcal L\left(V^{\otimes 2}\right)$ by 
\begin{equation}\label{cxrea4yqy}
Q(y_1,y_2)\,e_i\otimes e_j:=q(y_1,y_2,i,j)\,e_j\otimes e_i,\quad i,j\in\{1,\dots,r\}.\end{equation}
It is straightforward to see that $Q(y_1,y_2)$ is unitary, $Q^*(y_1,y_2)=Q(y_2,y_1)$, and $Q(y_1,y_2)$ satisfies the functional Yang--Baxter equation. Furthermore, $Q(y_1,y_2)$ satisfies Assumption~\ref{Rosated}, with 
\begin{equation}
\widetilde Q(y_1,y_2)e_i\otimes e_j=q(y_1,y_2,j,i)e_j\otimes e_i.\label{vcrte6ue}
\end{equation}

In accordance with formula  \eqref{ydsydq}, for the corresponding operator-valued distributions $A_i^+(x)$, $A_i^-(x)$ , as in \eqref{vftye6u} and \eqref{vctesay5ra5},  the following commutation relations hold:
\begin{align*}
A_i^+(x_1)A_j^+(x_2)&=q(y_2, y_1,i,j)A_j^+(x_2)A_i^+(x_1),\\
A_i^-(x_1)A_j^-(x_2)&=q(y_2,y_1,i,j)A_j^-(x_2)A_i^-(x_1),\\
A_i^-(x_1)A_j^+(x_2)&=\delta_{ij}\, \delta(y_1-y_2) +q(y_1,y_2,j,i)A_j^+(x_2)A_i^-(x_1)
\end{align*} 
for $i,j\in \{1,\dots,r\}$.

\begin{remark}\label{vcyrtds6u}
We may think of $q(y_1,y_2,i,j)$ as a continuous complex-valued function $q$ on $\big(Y\times \{1,\dots,r\}\big)^2$ which takes on values of modulus 1 and satisfies, for all $v_1,v_2\in Y\times \{1,\dots,r\}$, $\overline{q(v_1,v_2)}=q(v_2,v_1)$. Hence, the multicomponent system considered in this subsection can be thought of as a lifting  of the anyon commutation relations, for which the underlying space is $X\times \{1,\dots,r\}$ and the exchange function $Q$ is given by $q(y_1,y_2,j,i)$.  
\end{remark}

By Theorem~\ref{buty7e6u} (iii) and  \eqref{vcrte6ue}, a gauge-invariant quasi-free state $\tau$ is strongly quasi-free only in the trivial case where, for all $y_1,y_2\in Y$ and $i,j\in\{1,\dots,r\}$, we have $Q(y_1,y_2,i,j)=q_{ij}$ with  $q_{ij}\in\{-1,1\}$, and $q_{11}=\dots=q_{rr}$.

\begin{example}[Antiparticle]\label{vcte5y}
Let us consider the following example. Let $r=3$, and let a function $q(y_1,y_2,i,j)$ on $(Y\times\{1,2,3\})^2$ be defined as follows:  $q(y_1,y_2,1,1)=q(y_1,y_2,2,2)=q(y_1,y_2)$, $q(y_1,y_2,1,2)=q(y_1,y_2,2,1)=q(y_2,y_1)$ and $q(y_1,y_2,i,j)=1$ if at least one of the numbers $i$ and $j$ is 3. Here $q:Y^2\to\mathbb C$ is a continuous function, $|q(y_1,y_2)|=1$, and $q(y_2,y_1)=\overline{q(y_1,y_2)}$. Since $A_3^+(x_1)$ commutes with all $A_i(x_2)$ ($i=1,2,3$), the particle of type 3 is a boson. On the other hand, both particles of type 1 and 2 are  (abelian) anyons: $A_i^+(x_1)A_i^+(x_2)=q(y_2,y_1)A_i^+(x_2)A_i^+(x_1)$ ($i=1,2$). 

Recall the Introduction for an intuitive interpretation of fusion of quasiparticles. 
The fusion of particles of type 1 and 2 can be understood as a (heuristic) operator-valued distribution $B(x):=A^+_1(x)A^+_2(x)$. 
As easily seen, $B(x_1)$ also commutes with all $A_i(x_2)$ ($i=1,2,3$), and hence also with $B(x_2)$. Hence, one can identify $B(x)$ with $A_3(x)$. Thus the fusion of the particles of type 1 and $2$ gives a boson. This means that particle 2 is the antiparticle of particle 1, compare with \cite[Section~4.1.1]{Pachos}.
\end{example}

\begin{remark} It follows from Example~\ref{vcte5y} that, for the construction of a gauge-invariant quasi-free state on the algebra of the anyon commutation relations, one does doubling of the underlying space by adding to the particle its antiparticle. Indeed, in that case, the doubling of the space is equivalent to an addition of a particle of type 2, and the resulting commutation relations are $A_i^+(x_1)A_i^+(x_2)=Q(y_2,y_1)A_i^+(x_2)A_i^+(x_1)$ for $i=1,2$, and $A_i^+(x_1)A_j^+(x_2)=Q(y_1,y_2)A_j^+(x_2)A_i^+(x_1)$ for $i\ne j$.
\end{remark}

\subsection{Systems with particles of opposite type}
Let $V=\mathbb C^r$  ($r\ge2$), and just as above, let $\{e_1,\dots,e_r\}$ be the standard orthonormal basis in $V$. Let $\theta$ be a permutation from $S_r$ such that $\theta^2=\operatorname{id}$, i.e., all cycles in $\theta$ are of length one or two. If $\theta(i)=j$ and $i\ne j$, we may think of $i$ and $j$ as opposite types of particles.

\begin{theorem}\label{rtsw5uw5ude} For each $i,j\in\{1,\dots,r\}$, let $q(\cdot,\cdot,i,j)$ be 
 a continuous complex-valued function  on $Y^2$
 that satisfies the following assumptions, for all $y_1,y_2\in Y$ and $i,j\in\{1,\dots,r\}$,
 \begin{gather}
|q(y_1,y_2,i,j)|=1,\label{vctrsw5y32w}\\
\overline{q(y_1,y_2,i,j)}=q(y_2,y_1,i,j),\label{bhufg8ur8}\\
q(y_1,y_2,i,j)=q(y_1,y_2,j,i),\label{yft6u3wq7}\\
q(y_1,y_2,\theta(i),\theta(j))=q(y_1,y_2,i,j).\label{vuy6eu43wu}
\end{gather}
We also assume that there exists a $\varkappa\in\{-1,1\}$ such that 
\begin{equation}\label{gycdyrdsu}
q(y,y,i,i)=\varkappa\quad\text{for all $y\in Y$ and $i\in\{1,\dots,r\}$}.
\end{equation}
 For $(y_1,y_2)\in Y^2$, we define $Q(y_1,y_2)\in \mathcal L\left(V^{\otimes 2}\right)$ by 
\begin{equation}\label{ctrsw5y3w7q}
Q(y_1,y_2)\,e_i\otimes e_j:=q(y_1,y_2,i,j)\,e_{\theta(j)}\otimes e_{\theta(i)},\quad i,j\in\{1,\dots,r\}.\end{equation}
Then the following statements hold.

(i) $Q(y_1,y_2)$ is unitary, $Q^*(y_1,y_2)=Q(y_2,y_1)$, and the operator-valued function $Q(y_1,y_2)$ satisfies the functional Yang--Baxter equation. 

(ii) $Q(y_1,y_2)$ satisfies Assumption~\ref{Rosated}, with 
\begin{equation}\label{cdr5ws5qa2}
\widetilde Q(y_1,y_2)\,e_i\otimes e_j=q(y_1,y_2,\theta(i),j)\,e_{\theta(j)}\otimes e_{\theta(i)}=q(y_1,y_2,i,\theta(j))\,e_{\theta(j)}\otimes e_{\theta(i)}.
\end{equation}
The corresponding $Q$-MCR are given by 
\begin{align}
A_ i^+(x_1)A_ j^+(x_2)&= q(y_2,y_1, i,j)A_{\theta(j)}^+(x_2)A_{\theta(i)}^+(x_1),\notag\\
A_ i^-(x_1)A_  j^-(x_2)&=  q(y_2,y_1,  i,  j)A_{\theta(  j)}^-(x_2)A_{\theta(  i)}^-(x_1),\notag\\
A_  i^-(x_1)A_  j^+(x_2)&=\delta_{  i,\,  j}\,\delta(x_1-x_2)+  q(y_1,y_2,\theta(  i),  j)A_{\theta(  j)}^+(x_2)A_{\theta(  i)}^-(x_1).\label{gyrd6rde64ufr7rd}
\end{align}

(iii) For each appropriate choice of the operator $K$,  the corresponding gauge-invariant quasi-free state $\tau$ is strongly quasi-free if and only if 
\begin{equation}\label{ceay4q}
q(y_1,y_2,i,j)=q(y_2,y_1,\theta(i),j)\quad\text{for all }i,j\in\{1,\dots,r\}. \end{equation}
In the latter case, the commutation relation \eqref{gyrd6rde64ufr7rd}
becomes
$$A_  i^-(x_1)A_  j^+(x_2)=\delta_{  i,\,  j}\,\delta(x_1-x_2)+  q(y_2,y_1, i, j)A_{\theta(  j)}^+(x_2)A_{\theta(  i)}^-(x_1).$$
\end{theorem}

\begin{proof} (i) The statement easily follows from \eqref{vctrsw5y32w}--\eqref{vuy6eu43wu}.

(ii) Assumption~\ref{Rosated}~(i), (ii) and (iii) follow from \eqref{vctrsw5y32w}--\eqref{vuy6eu43wu} by straightforward calculations. Formula \eqref{gycdyrdsu} implies Assumption~\ref{Rosated} (iv). To show that  Assumption~\ref{Rosated}~(v) is satisfied, we note that the operator $\mathbf Q(y_1,y_2)$ acts as follows:
$$\mathbf Q(y_1,y_2)e_i\otimes e_j=\mathfrak q(y_1,y_2,i,j) e_{\theta(j)}\otimes e_{\theta(i)},$$
where 
\begin{equation}
\mathfrak q(y_1,y_2,i,j)=\begin{cases}
q(y_1,y_2,i,j),&\text{if either $y_1,y_2\in Y_1$ or $y_1,y_2\in Y_2$},\\
q(y_1,y_2,\theta(i),j)&\text{if either $y_1\in Y_1$, $y_2\in Y_2$ or $y_1\in Y_2$, $y_2\in Y_1$.}
\end{cases}
\label{vctsw5yw}
\end{equation}
 Formulas \eqref{vctrsw5y32w}--\eqref{vuy6eu43wu} and \eqref{vctsw5yw} imply that the functions $\mathfrak q(y_1,y_2,i,j)$  satisfy the formulas \eqref{vctrsw5y32w}--\eqref{vuy6eu43wu} in which $q(y_1,y_2,i,j)$ is replaced by $\mathfrak q(y_1,y_2,i,j)$, and $y_1,y_2\in\mathbf Y=Y_1\sqcup Y_2$. Now, the functional Yang--Baxter equation for $\mathbf Q(y_1,y_2)$ can be checked completely analogously to the proof of the  functional Yang--Baxter equation for $Q(y_1,y_2)$.
 
 The corresponding commutation relations hold in accordance with formula \eqref{ydsydq}.
 
 (iii) The statement follows immediately from Theorem~\ref{buty7e6u} (iii).  
\end{proof}

We will now consider some special cases of the operator-valued function $Q(y_1,y_2)$ as in Theorem~\ref{rtsw5uw5ude}.

\subsubsection{Two-component systems}\label{xrea5yw357}

Let $V=\mathbb C^2$. Let $\theta\in S_2$ be given by $\theta(1)=2$, $\theta(2)=1$. We fix two continuous functions $q_i: Y^2\to \mathbb C$ ($i=1,2$) such that, for all $(y_1,y_2)\in Y^2$, $\vert q_i(y_1,y_2)\vert =1$ and $\overline{q_i(y_1,y_2)}=q(y_2,y_1)$. We define $Q(y_1,y_2)\in \mathcal L(V^{\otimes 2})$ by 
\begin{align}
Q(y_1,y_2)e_i\otimes e_i&=q_1(y_1,y_2)e_{\theta(i)} \otimes e_{\theta(i)},\notag\\
Q(y_1,y_2)e_i\otimes e_{j}&=q_2(y_1,y_2)e_i \otimes e_{j},\quad i\ne j.\label{painkillerssss}
\end{align}
Thus, formula \eqref{ctrsw5y3w7q} holds with $q(y_1,y_2,i,i)=q_1(y_1,y_2)$  and $q(y_1,y_2,i,j)=q_2(y_1,y_2)$ if $i\not=j$. The assumptions~\eqref{vctrsw5y32w}--\eqref{gycdyrdsu} are obviously satisfied. 
The corresponding $Q$-MCR are then given by \eqref{utfdT}. 

Each corresponding gauge-invariant quasi-free state $\tau$ is strongly quasi-free if and only if $q_1(y_1,y_2)=q_2(y_2,y_1)=:q(y_1,y_2)$ for all $(y_1,y_2)\in Y^2$.
In the latter case, the $Q$-MCR are given by 
\begin{align}
A_i^+(x_1)A_i^+(x_2)&=q(y_2,y_1)A_{\theta(i)}^+(x_2)A_{\theta(i)}^+(x_1),\notag\\
 A_i^+(x_1)A_{j}^+(x_2)&=q(y_1,y_2)A_i^+(x_2)A_{j}^+(x_1),\quad i\ne j,\notag\\
A_i^-(x_1)A_i^-(x_2)&=q(y_2,y_1)A_{\theta(i)}^-(x_2)A_{\theta(i)}^-(x_1),\notag\\
 A_i^-(x_1)A_{j}^-(x_2)&=q(y_1,y_2)A_i^-(x_2)A_{j}(x_1),\quad i\ne j,\notag\\
A_i^-(x_1)A_i^+(x_2)&=\delta(x_1-x_2)+q(y_2,y_1)A_{\theta(i)}^+(x_2)A_{\theta(i)}^-(x_1),\notag\\
 A_i^-(x_1)A_{j}^+(x_2)&=q(y_1,y_2)A_i^+(x_2)A_{j}^-(x_1),\quad i\ne j,\notag
\end{align}
for $i,j\in\{1,2\}$.

\subsubsection{Three-component systems}
We will now construct a three-component system by adding an abelian anyon to the two-component from \S~\ref{xrea5yw357}. 

Let $V=\mathbb C^3$, and let $\theta\in S_3$ be given by $\theta(1)=2$, $\theta(2)=1$, $\theta(3)=3$. We fix four continuous functions $q_i: Y^2\to \mathbb C$ ($i=1,2,3,4$) such that, for all $(y_1,y_2)\in Y^2$, $\vert q_i(y_1,y_2)\vert =1$ and $\overline{q_i(y_1,y_2)}=q(y_2,y_1)$. We define $Q(y_1,y_2)\in \mathcal L(V^{\otimes 2})$ by \eqref{painkillerssss} for $i,j\in\{1,2\}$, and
\begin{align}
Q(y_1,y_2)e_3\otimes e_3&=q_3(y_1,y_2)e_{3} \otimes e_{3},\notag\\
Q(y_1,y_2)e_i\otimes e_3&=q_4(y_1,y_2)e_{3} \otimes e_{\theta(i)},\notag\\
Q(y_1,y_2)e_3\otimes e_i&=q_4(y_1,y_2)e_{\theta(i)} \otimes e_{3},\quad i=1,2.\notag
\end{align}
We also assume that $q_{1}(y,y)=q_3(y,y)=\varkappa$. The assumptions~\eqref{vctrsw5y32w}--\eqref{gycdyrdsu} are obviously satisfied.
The corresponding $Q$-MCR are  given by \eqref{utfdT} for $i,j\in\{1,2\}$. The commutation relations for the operators $A^+_3(x)$, $A^-_3(x)$ are the usual abelian anyon commutation relations governed by the function $q_3$. Finally, for $i\in\{1,2\}$,
\begin{align*}
A_i^+(x_1)A_3^+(x_2)&=q_4(y_2,y_1)A_3^+(x_2)A_{\theta(i)}^+(x_1),\\
A_i^-(x_1)A_3^-(x_2)&=q_4(y_2,y_1)A_3^-(x_2)A_{\theta(i)}^-(x_1),\\
A_i^-(x_1)A_3^+(x_2)&=q_4(y_1,y_2)A_3^+(x_2)A_{\theta(i)}^-(x_1),\\
A_3^-(x_1)A_i^+(x_2)&=q_4(y_1,y_2)A_{\theta(i)}^+(x_2)A_{3}^-(x_1).
\end{align*}

Each corresponding gauge-invariant quasi-free state $\tau$ is strongly quasi-free if and only if $q_1(y_1,y_2)=q_2(y_2,y_1)=:q(y_1,y_2)$ for all $(y_1,y_2)\in Y^2$, and both functions $q_3$ and $q_4$ are real-valued. In particular, this implies that the particles of type 3 must be either bosons or fermions.

\subsection{Fusion of quasiparticles}

Let $k\ge3$ be odd. We will now discuss fusion of $k$ quasiparticles described by the commutation relations~\eqref{utfdT}.

Recall that, for the quasiparticles discussed in  \S~\ref{xrea5yw357},
\begin{equation}\label{byrs6u53u}
A^+_{i}(x_1)A^+_{j}(x_2)=q(y_2,y_1,i,j)A^+_{\theta(j)}(x_2)A^+_{\theta(i)}(x_1),\quad i,j\in\{1,2\},
\end{equation}
where $\theta(1)=2$, $\theta(2)=1$ and 
\begin{equation}\label{vgydy6r75xx}
q(y_1,y_2,i,j)=
\begin{cases}
q_1(y_1,y_2),&\text{if }i=j,\\
q_2(y_1,y_2),&\text{if }i\not=j.
\end{cases}
\end{equation}

We start with a heuristic calculation of the corresponding exchange function $Q(y_1,y_2)$. 
To describe  $k$ quasiparticles at point $x$, we  define, for 
 $\mathbf i=(i_1,i_2,\dots,i_k)\in\{1,2\}^k$,
\begin{equation}\label{ve6e4wsa357}
A^+_{\mathbf i}(x):=A^+_{i_1}(x)A^+_{i_2}(x) \dotsm A^+_{i_k}(x).\end{equation}
(Note that the right hand side of \eqref{ve6e4wsa357} cannot be rigorously understood as an operator-valued distribution.) Since $k$ is an odd number, the commutation relation~\eqref{byrs6u53u} easily implies that, for any $x_1,x_2\in X$ and $\mathbf i,\mathbf j\in\{1,2\}^k$,
\begin{equation}\label{iyr7o4}
A^+_{\mathbf i}(x_1)A^+_{\mathbf j}(x_2)=\mathbf q(y_2,y_1,\mathbf i,\mathbf j)
A^+_{\theta(\mathbf j)}(x_2)A^+_{\theta(\mathbf i)}(x_1).
\end{equation}
Here, we denote $\theta(\mathbf i):=(\theta(i_1),\theta(i_2),\dots,\theta(i_k))$, and
\begin{equation}\label{yfxxdru6}
\mathbf q(y_1,y_2,\mathbf i,\mathbf j):=\prod_{l,m=1}^k q\big(y_1,y_2,\theta^{m-1}(i_l),\theta^{l-1}(j_m)\big).\end{equation} 

In view of the heuristic formula \eqref{iyr7o4}, we now proceed with a rigorous construction. Let $\{e_1,e_2\}$ be the standard orthonormal basis in $\mathbb C^2$, let $V=(\mathbb C^2)^{\otimes k}$, and we choose in $V$ the orthonormal basis $\left\{e_\mathbf i\mid\mathbf i\in\{1,2\}^k\right\}$, where   $e_\mathbf i:=e_{i_1}\otimes e_{i_2}\otimes\dots\otimes e_{i_k}$. We define $Q:Y^2\to\mathcal L(V^{\otimes 2})$ by
\begin{equation}\label{vcrt6u4}
Q(y_1,y_2)e_\mathbf i\otimes e_\mathbf j=\mathbf q(y_1,y_2,\mathbf i,\mathbf j)e_{\theta(\mathbf j)}\otimes e_{\theta(\mathbf i)},\end{equation}
where $\mathbf q(y_1,y_2,\mathbf i,\mathbf j)$ is defined by \eqref{yfxxdru6}. 

\begin{proposition}\label{vy6e683}
Assume $q_1(y,y)=q_2(y,y)=\varkappa\in\{-1,1\}$ for $y\in Y$. Then the functions $\mathbf q(y_1,y_2,\mathbf i,\mathbf j)$ satisfy the conditions of Theorem~\ref{rtsw5uw5ude}. For each appropriate choice of the operator $K$,  the corresponding gauge-invariant quasi-free state $\tau$ is strongly quasi-free if and only if $q_1(y_1,y_2)=q_2(y_2,y_1)=:q(y_1,y_2)$ for all $(y_1,y_2)\in Y^2$. 
\end{proposition}

\begin{proof}
The statement easily follows from the fact that the function $q(y_1,y_2,i,j)$ satisfies the conditions of Theorem~\ref{rtsw5uw5ude} and the definition~\eqref{yfxxdru6}. 
\end{proof}

\begin{remark} Let $Q:Y^2\to\mathcal L(V^{\otimes 2})$ be an arbitrary continuous function such that $Q(y_1,y_2)$ is a unitary operator in $V^{\otimes 2}$,    $Q^*(y_1,y_2)=Q(y_2,y_1)$, and $Q(y_1,y_2)$ satisfies the functional Yang--Baxter equation.  Let $k\ge 2$ and $W:=V^{\otimes k}$. Let $Q_k:Y^2\to\mathcal L(W^{\otimes 2})$ be a function derived from  fusion of $k$ quasiparticles with the exchange function $Q$. Then one can show that the continuous function $Q_k$ is also such that $Q_k(y_1,y_2)$ is a unitary operator in $W^{\otimes 2}$,    $Q^*(y_1,y_2)=Q(y_2,y_1)$, and $Q(y_1,y_2)$ satisfies the functional Yang--Baxter equation.   Note that, in the case where the operator $Q$ is as in \S~\ref{xrea5yw357}, we proved a stronger result: the operator-valued function $Q_k$  also satisfies Assumption~\ref{Rosated}, and if $Q$ satisfies \eqref{vutde7i}, then so does $Q_k$. 
\end{remark}

\subsubsection*{Acknowledgements}

This paper was mostly written when N.O. was a PhD student in the Department of Mathematics of Swansea University. N.O. is grateful to the Department for their constant support during his studies. 






\end{document}